\documentclass[11pt,english]{article}
\usepackage[utf8]{inputenc}
\usepackage[a4paper, total={6in, 9in}]{geometry}
\usepackage[pdftex]{graphicx}
\usepackage[table]{xcolor}
\usepackage{float}
\usepackage{hhline}
\usepackage{amssymb}
\usepackage{amsmath}
\usepackage{amsthm}
\usepackage{eurosym}
\usepackage{hyperref}
\usepackage{hhline}
\usepackage{soul}
\newtheorem{definition}{Definition}
\newtheorem{proposition}{Proposition}
\newtheorem{theorem}{Theorem}
\newtheorem{example}{Example}
\newtheorem{lemma}{Lemma}
\newtheorem{corollary}{Corollary}
\newcommand{\commentout}[1]{}

\title{Strengthening Nash Equilibria}
\author{Ivan Geffner \and Moshe Tennenholtz\thanks{The work by Ivan Geffner and Moshe Tennenholtz was supported by funding from
the European Research Council (ERC) under the European Union’s Horizon 2020
research and innovation programme (grant agreement 740435).}}
\date{}

\begin{document}

\maketitle

\begin{abstract}
Nash equilibrium is often heralded as a guiding principle for rational decision-making in strategic interactions. However, it is well-known that Nash equilibrium sometimes fails as a reliable predictor of outcomes, with two of the most notable issues being the fact that it is not resilient to collusion and that there may be multiple Nash equilibria in a single game. In this paper, we show that a mechanism designer can get around these two issues for free by expanding the action sets of the original game. More precisely, given a normal-form or Bayesian game $\Gamma$ and a Nash equilibrium $\vec{\sigma}$ in $\Gamma$, a mechanism designer can construct a new game $\Gamma^{\vec{\sigma}}$ by expanding the action set of each player and defining appropriate utilities in the action profiles that were not already in the original game. We show that the designer can construct $\Gamma^{\vec{\sigma}}$ in such a way that (a) $\vec{\sigma}$ is a semi-strong Nash equilibrium of $\Gamma^{\vec{\sigma}}$, and (b) $\vec{\sigma}$ Pareto-dominates or quasi Pareto-dominates all other Nash equilibria of $\Gamma^{\vec{\sigma}}$.
\end{abstract}

\newpage

\section{Introduction}

Suppose that three pirates want to share a treasure consisting of $300$ coins. Each pirate can individually choose if it cooperates or defects. If all three pirates cooperate, they split the treasure equally and each one of them gets $100$ coins. The same happens if only one pirate defects (she cannot overpower the other two), or if all three pirates defect. However, if exactly two pirates defect, they can keep the treasure for themselves and get $150$ coins each while the remaining pirate gets $0$. Clearly, all three pirates cooperating is a Nash equilibrium since no pirate gets additional utility by defecting. However, this equilibrium is not resilient to coalitions since two pirates can get additional utility by defecting together. Now, suppose that these three pirates are in an identical situation with the exception that there is an additional action available for them: \emph{Betray}. Betraying works as follows. If the number of pirates that play either \emph{Defect} or \emph{Betray} is different than two, nothing happens and all three pirates split the treasure equally. Otherwise, we have the following cases depending on what the two non-cooperative pirates play:
\begin{itemize}
    \item \textbf{(Defect, Defect)}: The two defecting pirates get $150$ coins each while the remaining pirate gets $0$.
    \item \textbf{(Betray, Defect)}: The pirate that betrays kills the pirate that defects and gets all $300$ coins.
    \item \textbf{(Betray, Betray)}: Both pirates kill each other. The remaining pirate gets all $300$ coins.
\end{itemize}

As in the previous game, the strategy profile in which all three pirates cooperate is a Nash equilibrium that is not resilient to collusion. Nevertheless, we can argue that in this scenario it is never in the interest of the pirates to collude. If two pirates decide to play \textbf{(Defect, Defect)}, one could easily betray the other in order to get a bigger share. In general, this type of Nash equilibrium is called \emph{semi-strong}~\cite{penn2005congestion}. In a semi-strong Nash equilibrium, no coalition can increase the utility of all of its members while simultaneously satisfying that no member of the coalition can further increase its utility by defecting (inside of the coalition). 

In the game with three pirates described above, we showed that the addition of the \emph{Betray} action converted a Nash equilibrium that was vulnerable to coalitions to a semi-strong one. A natural question that follows is if a mechanism designer can make any Nash equilibrium semi-strong by adding appropriate actions and utilities to the original game. In this paper, we give a positive answer to this question. Moreover, in normal-form games or Bayesian games with independent types, we show that we can add actions in such a way that the original equilibrium is not only semi-strong, but also Pareto-dominates all other Nash equilibria of the extended game. In arbitrary Bayesian games, we show that the original equilibrium quasi Pareto-dominates all other equilibria, which means that the total welfare of any subset $S$ of players with $|S| \ge 2$ is always greater in the original equilibrium. This implies that the mechanism designer can bluff about the existence of these additional actions since they would never be played in practice.

Perhaps not surprisingly, strengthening a pure Nash equilibrium $\vec{\sigma}$ is relatively simple. We can do so by simply adding a \emph{report} action to each player and defining the new utilities in such a way that players that defect from $\vec{\sigma}$ while being reported get punished (note that if two players play \emph{report}, they both are defecting while being reported), and players who report defecting players get rewarded. The real difficulty of the problem arises when $\vec{\sigma}$ is a mixed strategy, since in this case it is hard to tell if a player $i$ is defecting or if the action played by $i$ is a legitimate outcome from randomization. Additional complexity is added when players have (possibly correlated) private information and when the game is infinite (i.e., when the sets of possible actions are infinite). In fact, these cases require carefully calibrated mechanisms, which serve as the main technical contribution of our paper. At a high level, our approach is to allow players to \emph{bet} or \emph{predict} what actions or action profiles are going to be played by other players (besides playing an action in the original game). If they guess correctly, they get rewarded and the player whose action was predicted gets punished. If they guess wrong, they get punished. We show that we can calibrate the payoffs of such bets in such a way that, if a coalition defects and they all win, there exists a player in the coalition that has incentives to bet on other members of the coalition or there exists a player in the coalition that has incentives to modify its strategy because of the bets placed by other players.

\subsection{Related work}

The work in this paper intersects with two large branches of the economics and computer science literature. The first one is the study of \textbf{collusion-resilient protocols} in multi-agent systems. Ben-Or~\cite{ben1983another} and Rabin~\cite{rabin1983randomized} 
provided protocols that implement consensus in the presence of $t$ byzantine faults if the number of players $n$ satisfies $n > 3t$. These protocols were later generalized by Ben-Or, Goldwasser and Widgerson~\cite{bgw88} who showed that, if $n > 3t$, we can implement synchronous multiparty secure computation of any function $f$ on a finite domain in the presence of at most $t$ byzantine faults. Many subsequent authors added strategic incentives into the mix~\cite{abraham2011distributed,
aiyer2005bar,
abraham2008lower,
geffner2023communication,
vilacca2011n,
ranchal2021rational}. Bernheim, Peleg and Whinston~\cite{bernheim1987coalition} introduced a recursive definition for \emph{coalition-proof equilibrium}. Intuitively, a strategy profile $\vec{\sigma}$ is coalition-proof if there is no coalition that can defect from the main strategy to another coalition-proof strategy in such a way that all of its members get a better utility. Penn, Polukarov and Tennenholtz~\cite{penn2005congestion} defined \emph{semi-strong Nash equilibrium} in a similar way, but with only one level of recursion (i.e., for a coalition to succeed, their new strategy must also be a Nash equilibrium). Abraham, Dolev, Gonen and Halpern~\cite{abraham2006distributed} adapted Aumann's definition of \emph{$k$-resilience}~\cite{aumann1959acceptable} to cheap-talk games, and showed that every synchronous mediator can be implemented by a $k$-resilient cheap-talk strategy (i.e., a strategy in which coalitions of at most $k$ agents do not benefit from defecting from the proposed strategy) if $n > 2k$ and if there is an appropriate way to punish defecting players. Most of these results have been generalized to the asynchronous setting as well~\cite{bracha1984asynchronous, bracha1985asynchronous, ben1993asynchronous, geffner2023lower, geffner2021security}. In particular, Abraham, Dolev, Geffner and Halpern~\cite{abraham2019implementing} provided a $k$-resilient implementation of asynchronous mediators if $n > 3k$ and there exists a punishment strategy. Although most of the results in this area require that the ratio between the number of players $n$ and the maximum size $t$ or $k$ of the defecting coalition satisfies a certain threshold, our results show that, if we only need semi-strongness, we can take the protocol used for $t = 1$ or $k = 1$ and expand the set of actions in such a way that the same protocol tolerates deviations of coalitions of arbitrary size (according to the semi-strong definition). For example, this means that if a mechanism designer can modify (or can bluff about modifying) the action sets of the players, we provide a way to design semi-strong protocols that implement mediators in synchronous or asynchronous systems that require only three or four players, respectively.

The second closely related branch is \textbf{mechanism design}. The aim of a mechanism designer
is to design the set of strategies of the players and to define the utilities that they would get by following said strategies in order to influence the outcomes obtained from rational play. 
Sometimes, this influence can be exerted with the inclusion of a third-party mediator that may interfere with the game in different ways.
Aumann\cite{aumann1974subjectivity} considered the inclusion of a cheap-talk mediator in normal-form and Bayesian games. More precisely, given a normal-form game $\Gamma$, Aumann constructed a game $\Gamma_d$ where players can communicate via cheap-talk with a mediator $d$ before playing an action in $\Gamma$. The Nash equilibria of the resulting games were defined as correlated equilibria. Forges~\cite{forges1986approach} considered the same extension in Bayesian games and gave a geometric characterization of the resulting equilibria, which were labelled as \emph{communication equilibria}.  Monderer and Tennenholtz~\cite{monderer2009strong} showed that we can guarantee the existence of strong Nash equilibria in games with several classes of mediators. In their setting, some of these mediators can go beyond sending messages or signals to the players, such as for instance playing on their behalf when permitted. In another paper, Monderer and Tennenholtz~\cite{monderer2003k} considered a setting where the mediator has the power to \emph{promise payments} to the players (i.e., to increase their utility on given outcomes) but cannot communicate. They study the minimum payment that is necessary to influence the behavior of the players. Monderer and Tennenholtz showed that, given any normal-form game $\Gamma$ and any pure Nash equilibrium $\vec{\sigma}$ in $\Gamma$, there exists a mediator $d$ that promises payments only outside of the equilibrium path in such a way that $\vec{\sigma}$ is not only a Nash equilibrium but also a dominant strategy in the newly generated game. Tennenholtz~\cite{tennenholtz2004program} later considered a type of normal-form games where players are programs and a third-party mediator is able to access and reveal the source code of other players before they commit to an action. Tennenholtz showed that the set of Nash equilibria in this setting (or \emph{program equilibria}, as defined in \cite{tennenholtz2004program}) is richer than without the inclusion of the mediator. For example, there exist program equilibria in the one-shot Prisoner's Dilemma where both players cooperate. Oesterheld et al.~\cite{oesterheld2022similarity} refined this idea and defined \emph{similarity-based coopaeration}, where each player only has access to a similarity parameter instead of the whole source code of other players. They showed that, in their setting, there also exist equilibria in which both players cooperate in the Prisoner's Dilemma. Our setting can also be viewed as including a mediator into a normal-form or Bayesian game in order to influence the outcome. In this case, the mediator cannot communicate with the players but can extend the game by adding new actions with arbitrary utilities.

\subsection{Paper Structure}

The rest of the paper is structured as follows. In Section~\ref{sec:definitions} we give the necessary definitions to formalize and prove our main results, which are stated in Section~\ref{sec:results}. Section~\ref{sec:sketch} provides a high-level overview of the proofs, along with several examples that showcase the main difficulties that are needed to overcome in each setting. We give a formal proof of our main results in Sections~\ref{sec:proof-thm1} and \ref{sec:arbitrary-bayesian}. The first of these sections deals with the case of finite normal-form games and finite Bayesian games with independent types, and the latter deals with (arbitrary) finite Bayesian games. These proofs are generalized in Section~\ref{sec:infinite} in order to hold for games with infinite actions as well. We give a conclusion in Section~\ref{sec:conclusion}, along with a brief summary about the remaining open problems.

\section{Basic definitions}\label{sec:definitions}

\subsection{Games and Equilibria}

A normal-form game is a tuple $(P, A, U)$ where $P = \{1,2,\ldots,n\}$ is the set of players, $A := A_1 \times A_2 \times \ldots \times A_n$ is the set of possible actions profiles and $U = (u_1, u_2, \ldots, u_n)$ is a tuple of utility functions such that $u_i: A \times \mathbb{R}$ outputs the utility of player $i$ given the action profile played. In a normal form game, a (mixed) strategy $\sigma_i$ for player $i$ is simply a distribution of actions of $A_i$ (we denote by $\Delta A_i$ the set of such possible distributions). A strategy profile $\vec{\sigma} = (\sigma_1, \ldots, \sigma_n)$ is a Nash Equilibrium of $\Gamma$ if no player can increase its utility by switching to another strategy. More precisely, $\vec{\sigma}$ is a Nash equilibrium if $u_i(\vec{\sigma}) \ge u_i(\vec{\sigma}_{-i}, \tau_i)$ for all players $i$ and all strategies $\tau_i$ for $i$ (where $\vec{\sigma}_{-i}$ denotes that all players but $i$ play their part of strategy profile $\vec{\sigma}$).

A Bayesian game is a generalization of normal-form games where players have private types and the utility of each player depends not only on the action profile played but also on its own type. More formally, a Bayesian game is a tuple $(P,T,q, A, U)$ where $T = T_1 \times T_2 \times \ldots \times T_n$ is the set of possible type profiles and $q$ is the prior distribution over $T$ that is common knowledge among the players. $P$, $A$ and $U$ denote the same concepts as in normal-form games, except that the domain of each utility function $u_i$ is $T_i \times A$ instead of just $A$. In Bayesian games, strategies are functions from types to distributions over action profiles. More precisely, in a Bayesian game $\Gamma = (P,T,q,A,U)$, a strategy $\sigma_i$ for player $i$ maps elements from $T_i$ to distributions over $A$. As in normal-form games, a strategy profile $\vec{\sigma}$ is a Nash equilibrium if $u_i(\vec{\sigma}) \ge u_i(\vec{\sigma}_{-i}, \tau_i)$ for all players $i$ and all strategies $\tau_i$ for $i$. For simplicity, we will assume that no player has a type that has $0$ probability of being sampled.

Given a normal-form or Bayesian game $\Gamma$, we say that $\Gamma$ is \emph{finite} if the set of actions is finite, and we say that $\Gamma$ is \emph{bounded} if all utility functions are bounded. Note that if a game is finite, it is also bounded.


\commentout{
\begin{equation}\label{eq:expected-payoff}
u_i(\vec{\sigma}) = \sum_{t_i \in T_i} q(t_i) \sum_{\vec{t}}
q(\vec{t} \mid t_i) u_i(\vec{\sigma}(\vec{t})),
\end{equation}
where 
$u_i(\vec{\sigma}(\vec{t}))$ denotes the expected utility of player
$i$ when an action profile 
is chosen
according to 
$(\sigma(\vec{t}))$. In Bayesian games, Nash equilibria are defined in the same way as in normal form games: a strategy profile $\vec{\sigma}$ is a (Bayesian) Nash equilibrium if $u_i(\vec{\sigma}) \ge u_i(\vec{\sigma}_{-i}, \tau_i)$ for all players $i$ and all strategies $\tau_i$ for $i$.
}

\subsection{Resilience to coalitions}

There are well-known notions of equilibria that are resilient to collusion and joint deviations. The first of such notions is \emph{strong Nash equilibrium}~\cite{aumann1959acceptable}, which occurs if no possible coalition can strictly increase the payoff of all its members by jointly deviating. More precisely, in a normal-form or Bayesian game $\Gamma = (P,T,q,A,U)$, a strategy profile $\vec{\sigma}$ is a strong Nash equilibrium if, for all subsets $S \subseteq P$ and all strategy profiles $\vec{\tau}_S$ for players in $S$, there exists a player $i \in S$ such that $u_i(\vec{\sigma}) \ge u_i(\vec{\sigma}_{-S}, \vec{\tau}_S)$.

In this paper we focus on a variant of strong Nash equilibrium introduced by Penn, Polukarov and Tennenholtz~\cite{penn2005congestion} in which, for a joint deviation to be successful, not only all members must strictly benefit from colluding, but also the strategy performed by the members of the coalition must  be a Nash equilibrium inside of the larger game (in which the players outside of the coalition simply play their part of $\vec{\sigma}$). The following definition makes this precise.

\begin{definition}\label{def:semi-strong}
Let $\Gamma$ be a normal-form or a Bayesian game 
and $\vec{\sigma}$ be a strategy profile for $\Gamma$. Then, $\vec{\sigma}$ is a semi-strong Nash equilibrium if, for all subsets $S$ and all strategy profiles $\vec{\tau}_S$ for $S$, there exists a player $i \in S$ and a strategy $\tau'_i$ for $i$ such that $$u_i(\vec{\sigma}) \ge u_i(\vec{\sigma}_{-S}, \vec{\tau}_S)\quad \mbox{or} \quad u_i(\vec{\sigma}_{-S}, \vec{\tau}_{S \setminus \{i\}}, \tau'_i) > u_i(\vec{\sigma}_{-S}, \vec{\tau}_S)$$
\end{definition}

Note that if the first condition is satisfied for some $i$, then $\vec{\tau}_S$ is not profitable for $i$, while the second condition implies that $i$ can further deviate from $\vec{\tau}_S$ to increase its utility from the one it would get in the coalition.

\subsection{Game extensions}

The results of this paper involve games that result from extending the set of actions of simpler games. We say that a Bayesian game $\Gamma' = (P', T', q', A', U')$ \emph{extends} another Bayesian game $\Gamma = (P,T,q,A,U)$ if $P' = P$, $T' = T$, $q' = q$, $A_i \subseteq A'_i$ for all $i \in P$, and $u_i(\vec{a}) = u'_i(\vec{a})$ for all $\vec{a} \in A$ and $i \in P$. Intuitively, game extensions must have the same players, action profiles, and utilities as the original game, but they may also include additional actions for the players that may give arbitrary utilities when played. The definition for normal-form games is analogous.

For future reference, note that if $\Gamma'$ extends $\Gamma$, we can view strategy profiles of $\Gamma$ as strategy profiles in $\Gamma'$ in a natural way: each player $i$ plays each action in $A_i$ with the same probability as in $\Gamma$, and plays each action in $A'_i \setminus A_i$ with probability $0$. Unless it is specified, for the rest of the paper we will always use such convention.

\subsection{Comparing Nash Equilibria}

Given a Bayesian or Normal-form game $\Gamma$ and two Nash equilibria $\vec{\sigma}$ and $\vec{\tau}$, we say that $\vec{\sigma}$ \emph{Pareto-dominates} $\vec{\tau}$ if $u_i(\vec{\sigma}) \ge u_i(\vec{\tau})$ for all $i \in P$. Moreover, we say that $\vec{\sigma}$ \emph{quasi Pareto-dominates} $\vec{\tau}$ if, for all subsets $S \subseteq P$ with $|S| \ge 2$, $$\sum_{i \in S} u_i(\vec{\sigma}) \ge \sum_{i \in S} u_i(\vec{\tau}).$$

Note that, if $|P| \ge 2$ and $\vec{\sigma}$ quasi Pareto-dominates $\vec{\tau}$, 
\begin{itemize}
    \item At most one player $i$ satisfies $u_i(\vec{\tau}) > u_i(\vec{\sigma})$.
    \item The total welfare of the players is greater with $\vec{\sigma}$ than with $\vec{\tau}$.
\end{itemize}

\section{Main Results}\label{sec:results}

The main contribution of this paper can be summarized by the following Theorems.

\begin{theorem}\label{thm:main}
Let $\Gamma = (P,T,q,A,U)$ be a bounded Bayesian game with independent types (i.e., that $q$ samples the types independently). 
Then, for each Nash equilibrium $\vec{\sigma}$ of $\Gamma$ there exists an extension 
$\Gamma^{\vec{\sigma}} = (P, T, q, A^{\vec{\sigma}}, U^{\vec{\sigma}})$
of $\Gamma$ such that:
\begin{itemize}
    \item [(a)] $\vec{\sigma}$ is a semi-strong Nash equilibrium of $\Gamma^{\vec{\sigma}}$.
    \item [(b)]
    $\vec{\sigma}$ Pareto-dominates all other Nash equilibria of $\Gamma^{\vec{\sigma}}$.
\end{itemize}

Moreover, if $A$ is finite, $A^{\vec{\sigma}}$ satisfies that $|A^{\vec{\sigma}}_i| \le |A_i|\cdot |P| \cdot |A|_{max}$ for all $i \in P$, where $|A|_{max} = \max_{i \in P}{|A_i|}$.
\end{theorem}

Intuitively, part (a) of Theorem~\ref{thm:main} states that given any game and any Nash equilibrium, we can expand the game such that the original Nash equilibrium becomes semi-strong. At the same time, part (b) implies that none of the newly added actions will ever be played if all players play optimally. The last part of the Theorem states that, if $\Gamma$ is finite, then the size of the resulting extension is polynomial over the size of $\Gamma$. The following Theorem extends Theorem~\ref{thm:main} to arbitrary Bayesian games.

\begin{theorem}\label{thm:main2}
Let $\Gamma = (P,T,q,A,U)$ be a bounded Bayesian game.
Then, for each Nash equilibrium $\vec{\sigma}$ of $\Gamma$ there exists an extension 
$\Gamma^{\vec{\sigma}} = (P, T, q, A^{\vec{\sigma}}, U^{\vec{\sigma}})$
of $\Gamma$ such that:
\begin{itemize}
    \item [(a)] $\vec{\sigma}$ is a semi-strong Nash equilibrium of $\Gamma^{\vec{\sigma}}$.
    \item [(b)]
    $\vec{\sigma}$ quasi Pareto-dominates all other Nash equilibria of $\Gamma^{\vec{\sigma}}$.
\end{itemize}

Moreover, if $A$ is finite, $A^{\vec{\sigma}}$ satisfies that $|A^{\vec{\sigma}}_i| = |A_i| + |A|$ for all $i \in P$.
\end{theorem}

The only two differences between Theorem~\ref{thm:main} and Theorem~\ref{thm:main2} are the size of the extension and the fact that, in Theorem~\ref{thm:main2}, the resulting equilibrium quasi Pareto-dominates all other Nash equilibria of $\Gamma^{\vec{\sigma}}$ instead of Pareto-dominating. 
Note that Theorem~\ref{thm:main} and Theorem~\ref{thm:main2} do not require that $\Gamma$ is finite. In fact, as shown in Section~\ref{sec:infinite}, the statements still hold when the action sets are uncountable, as long as $\Gamma$ is bounded.

\section{Sketch of the Proofs}\label{sec:sketch}

As in the \emph{three pirates} example from the introduction, the idea behind 
the construction of $\Gamma^{\sigma} = (P, A^{\vec{\sigma}}, U^{\vec{\sigma}})$ is that players can play special actions that reward them if another player goes off the equilibrium path. This construction can be very simple for pure Nash equilibria: suppose that $\vec{\sigma}$ is a pure Nash equilibrium where each player $i$ plays some action $a_i$ with probability 1. Then, we can set $A^{\vec{\sigma}}_i = A_i \cup \{\alpha \}$ and define the payoffs as follows when some player $i$ plays $\alpha$: 
\begin{itemize}
    \item If another player $j$ plays $\alpha$, all players $k$ that didn't play $a_k$ get $m-1$ utility (including those who played $\alpha$), where $m$ is the minimum utility of the game. All other players get the same utility as in equilibrium.
    \item If no other player plays $\alpha$, 
    \begin{itemize}
        \item If all other players $j$ play $a_j$, they all get the same utility as in equilibrium and $i$ gets $m-1$.
        \item Otherwise, $i$ gets the maximum possible utility plus $1$, all players $j$ that played $a_j$ get the same utility as in equilibrium, and all other players get $m-1$. 
    \end{itemize}
\end{itemize}

It is straightforward to check that $\Gamma^{\vec{\sigma}}$ satisfies the conditions of Theorem~\ref{thm:main}: First, no player can increase its utility by defecting from $\vec{\sigma}$. Second, if a coalition defects from $\vec{\sigma}$ and they all increase their utility, it means that none of them is playing $\alpha$. However, this implies that a player can further increase its utility by switching its action to $\alpha$. Lastly, an analogous reasoning shows that $\vec{\sigma}$ is Pareto-optimal: if another strategy profile $\vec{\tau}$ gives a better utility than $\vec{\sigma}$ to all players, it means that none of them is playing $\alpha$. Then, again, this would not be a Nash equilibrium since a player can increase its utility by playing $\alpha$ instead.

If we consider mixed Nash equilibria, the construction becomes more complicated. For instance, there could be Nash equilibria where all actions have non-zero chance of being played, as in the following example. 

\begin{example}
Consider a normal-form game $\Gamma = (P, A, U)$ with $P = \{1,2,\ldots, n\}$, $A_i = \{0,1\}$ for all $i \in P$, and such that $$u_i(a_1, \ldots, a_n) = 
\left\{
\begin{array}{ll}
1 & \mbox{if } \#a_i(a_1, \ldots, a_n) \ge \#(1-a_i)(a_1, \ldots, a_n) \\
0 & \mbox{otherwise,}
\end{array}
\right.
$$
where $\#a_i(a_1, \ldots, a_n)$ denotes the number of players who played action $a_i$. Denote by $\vec{\sigma}$ the strategy profile in which each player chooses its action uniformly at random.
\end{example}

Intuitively, $i$ gets $1$ utility if its action is played by the majority of the players and $i$ gets $0$ utility otherwise. Clearly, if all other players randomize their action, it doesn't matter what action $i$ chooses since both give the same utility in expectation. Therefore, randomizing its own action is also a best response for $i$, which means that $\vec{\sigma}$ is a Nash equilibrium. However, $\vec{\sigma}$ is not semi-strong: if two or more players collude, it is in their interest to vote on the same value. For instance, all of them playing $0$ or all of them playing $1$ strictly increases their expected utility and no agent in the coalition can further increase it by playing anything else. 

In order to create an extension of $\Gamma$ such that $\vec{\sigma}$ is a semi-strong Nash equilibrium we need a different approach than the one given for pure Nash equilibria. The main issue is that it is not easy to tell if a player $i$ is defecting from the proposed strategy or not since, regardless of what $i$ plays, its action could be the result of playing $\sigma_i$. They key idea is that, in equilibrium, $i$ must play $0$ and $1$ with $\frac{1}{2}$ probability each. By contrast, if $i$ defects, one of these actions must be played with probability greater than $\frac{1}{2}$. We can make use of this fact as follows. Suppose that, for each player $i$ and each action $a_i \in \{0,1\}$, we give each player $j$ the possibility of \emph{betting} that $i$ will play $a_i$. If $i$ plays $a_i$, then $j$ gets $x$ utility and $i$ loses $x'$ utility. Otherwise, $j$ gets $-x$ utility. Clearly, if $i$ follows ${\sigma}_i$, there is no incentive for $j$ to place such a bet. However, if $i$ defects, betting on one of $i$'s actions gives $j$ a strictly positive expected utility. We can adjust the values of $x$ and $x'$ such that it is never in the interest of $i$ to collude and defect. The construction given in Section~\ref{sec:construction} follows a similar approach, except for the fact that the utility won by winning a bet and the utility lost by losing a bet may be different. More precisely, the bets are calibrated in such a way that, if $i$ plays some action $a$ with the same probability as in $\sigma_i$, the expected payoff of betting on $(i,a)$ is $0$, while if $i$ plays $a$ with more probability than in $\sigma_i$ the expected payoff of betting on $(i,a)$ is strictly positive.

The construction of $\Gamma^{\vec{\sigma}}$ for Bayesian games is more complicated. The main issue is that, in normal-form games, whenever a player $i$ defects from $\sigma_i$ to a different strategy $\tau_i$, it is guaranteed that there exists an action that is more likely to be played with $\tau_i$ than with $\sigma_i$. However, this is not true for Bayesian games. For example, consider a game such that $T_i = A_i = \{0,1\}$ for all players $i$, and such that all types are sampled independently and uniformly from $\{0,1\}$. If $\sigma_i$ is the strategy in which player $i$ plays whatever its type is and $\tau_i$ is the strategy in which it plays the opposite of its type, it is indistinguishable to the other players if $i$ is playing $\sigma_i$ or if $i$ is playing $\tau_i$. In particular, we cannot design bets that give a strictly positive expected utility when $i$ plays $\tau_i$ and simultaneously give $0$ expected utility when $i$ plays $\sigma_i$. The key observation is that, if $\vec{\tau}_S$ is a profitable deviation for some coalition $S$ and types are independent, each player $i \in S$ must be able to tell the strategy $\tau_j$ of at least one other player $j \in S$. Otherwise, $i$ would get the same utility in $(\vec{\sigma}_{-i}, \tau_i)$ (see Proposition~\ref{prop:distinguishable1}). This observation implies that the construction used for normal-form games can be easily generalized to also cover Bayesian games with independent types.

Unfortunately, the construction used for normal-form games does not quite work in general for Bayesian games, as the following example shows.

\begin{example}\label{example:indistinguishable}
Consider a Bayesian game $\Gamma = (P, T, q, A, U)$ such that $P = \{1,2,3\}$, $T_i = A_i = \{0,1\}$ for all $i \in P$, and $q$ is the uniform distribution over the subset of $T$ such that the sum of coordinates is odd (i.e., $\{(t_1,t_2,t_3) \in T \mid t_1 + t_2 + t_3 \equiv 1 \bmod 2\}$). Given an action profile $(a_1, a_2, a_3)$, all players get $1$ utility if $a_1+a_2+a_3$ is odd. Otherwise, they all get $0$ utility.

Consider two strategy profiles $\vec{\sigma}$ and $\vec{\tau}$. In $\vec{\sigma}$, each player independently plays $0$ and $1$ with equal probability. In $\vec{\tau}$, each player plays its own type.
\end{example}

In Example~\ref{example:indistinguishable}, it is straightforward to check that $\vec{\sigma}$ and $\vec{\tau}$ are both Nash equilibria where each player plays $0$ and $1$ with equal probability. However, players get $1/2$ expected utility with $\vec{\sigma}$ and $1$ expected utility with $\vec{\tau}$. The main difference between the strategy profiles is that, even though each player has $1/2$ chance of having type $0$ and $1/2$ chance of having type $1$, their types are correlated in such a way that their sum is always an odd number. Moreover, we can easily check that players' types are \textbf{pairwise independent}. This means that the probability that a player $i$ plays a certain action $a \in \{0,1\}$ with $\tau_i$ is exactly the same as the probability that $i$ plays $a$ with $\sigma_i$, even if we condition the probability on another player $j$'s type. Therefore, $j$ cannot guess $i$'s action with $\vec{\tau}$ any better than with $\vec{\sigma}$, which implies that the construction used for normal-form games fails in this case. 

In order to deal with situations as the one presented in Example~\ref{example:indistinguishable}, we modify the construction $\Gamma^{\vec{\sigma}}$ used for normal form games in such a way that players can bet on action profiles played by other players instead of individual actions. It can be shown that if bets are calibrated in a similar fashion, $\Gamma^{\vec{\sigma}}$ satisfies the properties of Theorem~\ref{thm:main2}.

To extend the proofs of Theorems~\ref{thm:main} and \ref{thm:main2} to infinite games, we need a slightly different construction, as shown by the following example.

\begin{example}\label{example:infinite}
Consider a normal-form game $\Gamma = (P, A, U)$ with $P = \{1,2\}$, $A_1 = A_2 = \mathbb{N}$, and $$u_1(a_1, a_2) = u_2(a_1, a_2) = \left\{
\begin{array}{ll}
0 & \mbox{ if } a_1 = 0 \mbox{ or } a_2 = 0\\
1 & \mbox{ otherwise.}
\end{array}
\right.$$

Denote by $\vec{\sigma}$ the strategy profile in which each player plays $0$ with probability 1, and by $(\vec{\tau})_N$ the strategy profile in which both players sample their action uniformly at random from $\{1,2, \ldots, N\}$.
\end{example}

It is easy to check that $\vec{\sigma}$ and $(\vec{\tau})_N$ are Nash equilibria of $\Gamma$. Suppose that we construct $\Gamma^{\vec{\sigma}}$ in the same way as in normal-form games and have each player lose $C$ utility whenever the other player guesses its action correctly. By playing $\vec{\sigma}$, each player gets $0$ expected utility. However, if they play $(\vec{\tau})_N$, each player has at most $1/N$ chance of guessing the other player's action correctly. Therefore, regardless of how they place their bets, each player can get at least $1 - \frac{C}{N}$ expected utility in $\Gamma^{\vec{\sigma}}$. If $N$ is large enough, this value is larger than the utility that they would get with $\vec{\sigma}$. 

The main issue presented in Example~\ref{example:infinite} is that players can spread their actions really thin, avoiding this way being caught by other players. Even though this doesn't always lead to a better equilibrium (e.g., it may happen that no matter how player $i$ spreads her actions, it is always better for $i$ to spread them even thinner), this situation is clearly non-desirable. In order to avoid it, the construction in Section~\ref{sec:infinite} allows players to bet on arbitrary sets of actions (or arbitrary sets of action profiles when dealing with Bayesian games) instead of just individual actions. Again, these bets are calibrated in such a way that no player benefits from betting if all other players play their part of $\vec{\sigma}$. However, if a player $i$ defects from $\vec{\sigma}$, betting on the set of actions that $i$ is more likely to play after defecting always gives a strictly positive payoff. It can be shown that the rewards and punishments from betting can be calibrated in such a way that it is always better to follow the proposed equilibrium $\vec{\sigma}$.

\section{Proof of Theorem~\ref{thm:main} for Finite Games}\label{sec:proof-thm1}

In this section we prove Theorem~\ref{thm:main} for the case in which $\Gamma$ is finite. We extend this proof for infinite games in Section~\ref{sec:infinite}.

In the following sections, we prove Theorem~\ref{thm:main} for finite normal-form games. Later, in Section~\ref{sec:bayesian-construction}, we extend the proof to Bayesian games with independent types.

\subsection{Construction of $\Gamma^{\sigma}$}\label{sec:construction}


In this section we give the precise construction of $\Gamma^{\vec{\sigma}}$. Define $$A_i^{\vec{\sigma}} := A_i \cup \left(A_i
 \times \bigcup_{j \in P} \left(\{j\} \times A_j\right)\right).$$ Elements of $A_i^{\vec{\sigma}}$ are either actions of $A_i$ or a tuple $(a_i, j, a_j)$ that indicates that $i$ is playing action $a_i$ but is also betting that player $j$ is playing action $a_j$. Suppose that players play actions $(a_1, b_1), (a_2, b_2), \ldots, (a_n, b_n)$ where, for convenience, we use the notation $b_i := \bot$ to denote that player $i$ placed no bet (i.e., if $i$ played an action in $A_i$). The utility $u^{\vec{\sigma}}_i$ of player $i$ is given by
 $$\begin{array}{lll}
 u^{\vec{\sigma}}_i((a_1, b_1), (a_2, b_2), \ldots, (a_n, b_n)) &  :=  & u_i(a_1, a_2, \ldots, a_n) \\
  & + & w_i^{\vec{\sigma}}((a_1, b_1), (a_2, b_2), \ldots, (a_n, b_n))\\
  & + & \ell_i^{\vec{\sigma}}((a_1, b_1), (a_2, b_2), \ldots, (a_n, b_n)),
  \end{array}
  $$

where $w_i$ is the additional utility that $i$ gets by placing a bet on another player's action, and $\ell_i$ is the additional (negative) utility that $i$ gets if another player bets on $i$'s action and wins. If $p_{(j, a_j)}^{\vec{\sigma}}$ is the probability that strategy $\sigma_j$ assigns to action $a_j$, then $$w_i^{\vec{\sigma}}((a_1, b_1), (a_2, b_2), \ldots, (a_n, b_n)) := 
\left\{
\begin{array}{ll}
0 & \mbox{if } b_i = \bot\\
1 - p_{(j, a_j)}^{\vec{\sigma}} & \mbox{if } b_i \mbox{ is of the form } (j, a_j)\\
- p_{(j, a_j)}^{\vec{\sigma}} & \mbox{if } b_i \mbox{ is of the form } (j, a) \mbox{ with } a \not= a_j.\\
\end{array}
\right.$$
It is important to note that, if $j$ does not defect and $i$ bets on $a_j$, then its expected additional utility is given by $p_{(j, a_j)}^{\vec{\sigma}}\left(1 - p_{(j, a_j)}^{\vec{\sigma}}\right) + \left(1 - p_{(j, a_j)}^{\vec{\sigma}}\right)\left(- p_{(j, a_j)}^{\vec{\sigma}}\right) = 0$.

Finally, we define $\ell_i^{\vec{\sigma}} := \sum_{j \in P} \ell_{(i,j)}^{\vec{\sigma}}$, where $$
  \ell_{(i,j)}^{\vec{\sigma}}((a_1, b_1), (a_2, b_2), \ldots, (a_n, b_n)) := 
\left\{
\begin{array}{ll}
0 & \mbox{if } b_j = \bot \mbox{ or } b_j \mbox{ is of the form } (i, a) \mbox{ with } a \not= a_i\\
-C & \mbox{otherwise.}
\end{array}\right.$$

The constant $C$ used in the previous definition is simply a large value that serves as a deterrent for players to defect from $\vec{\sigma}$: if a player $i$ follows a different strategy $\tau_i$, then there exists some action $a \in A_i$ such that  $a$ is played with strictly more probability in $\tau_i$ than in $\sigma_i$. If this is the case, because of the way in which we defined $w_j$, it is in the interest of other players $j$ to bet on $(i, a_i)$, which decreases $i$'s overall utility. In Section~\ref{sec:proof}, we show that we can choose a value for $C$ so that conditions (a) and (b) of Theorem~\ref{thm:main} are satisfied.

\subsection{Proof of Correctness}\label{sec:proof}

\subsubsection{Part (a)}

We start by showing that $\vec{\sigma}$ is a semi-strong Nash equilibrium if $\Gamma^{\vec{\sigma}}$. First, as observed in the previous section, it is easy to check that our construction of $\Gamma^{\vec{\sigma}}$ satisfies the following property:

\begin{lemma}\label{lemma:bets-basic}
Given a strategy profile $\vec{\tau}$ for $\Gamma^{\vec{\sigma}}$, a player $j \in P$, and an action $a_j \in A_j$, let $p^{\vec{\tau}}_{(j,a_j)}$ denote the probability that $j$ plays an action in $A_j^{\vec{\sigma}}$ such that its first component is $a_j$. Let $x_{(j,a_j)}^{\vec{\tau}}$ denote the expected additional utility of any other player when betting that $j$ plays action $a_j$. Then
$$x_{(j,a_j)}^{\vec{\tau}} = p^{\vec{\tau}}_{(j,a_j)} - p^{\vec{\sigma}}_{(j,a_j)}$$
\end{lemma}

\begin{proof}
The expected additional utility can be computed as follows:
$$x_{(j,a_j)}^{\vec{\tau}} = p^{\vec{\tau}}_{(j,a_j)}\left(1 - p^{\vec{\sigma}}_{(j,a_j)}\right) + \left(1 - p^{\vec{\tau}}_{(j,a_j)}\right)\left(-p^{\vec{\sigma}}_{(j,a_j)}\right) =  p^{\vec{\tau}}_{(j,a_j)} - p^{\vec{\sigma}}_{(j,a_j)}.$$
\end{proof}

With this, it is straightforward to check that, if $C \ge 1$, $\vec{\sigma}$ is a (standard) Nash equilibrium of $\Gamma^{\vec{\sigma}}$: Fix a player $i$ and assume all other players play $\vec{\sigma}_{-i}$. Since $\vec{\sigma}$ is a Nash equilibrium of $\Gamma$, regardless of the bets $i$ places, it is a best response for $i$ to distribute the first component of its action in the same way as $\sigma_i$. It just remains to check that $i$ gets no utility from placing bets. By Lemma~\ref{lemma:bets-basic}, $i$ gets no utility from betting on other players. Moreover, if $C \ge 1$, $i$ gets negative expected utility by betting on itself. This shows that $\sigma_i$ is indeed a best response for $i$.

Now suppose that $\vec{\sigma}$ does not satisfy the conditions of 
Definition~\ref{def:semi-strong}. This would imply that there exists a 
coalition $S \subseteq P$ and a strategy profile $\vec{\tau}_S$ for players
in $S$ such that (a1) all players $i \in S$ get a strictly better expected 
utility than $u_i(\vec{\sigma})$, and (b1) no player $i \in S$ is better off 
defecting from $\vec{\tau}_S$, assuming that the rest of the players play 
$(\vec{\sigma}_{-S}, \vec{\tau}_{S \setminus\{i\}})$. Given a strategy 
$\tau_i$ for some player $i \in P$ in $\Gamma^{\vec{\sigma}}$, denote by 
$\tilde{\tau}_j$ the strategy for $i$ in $\Gamma$ obtained by projecting the
actions of $A_j^{\vec{\sigma}}$ onto their first component (i.e., by getting 
rid of the bets). Then, we have the following:

\begin{proposition}\label{prop:different-strat}
Let $M$ be the difference between the maximum and minimum utility of $\Gamma$. If $S$ and $\vec{\tau}_S$ satisfy (a1), there exist at least one player $i \in S$ such that $\tilde{\tau}_i \not = \sigma_i$.
\end{proposition}

\begin{proof}

If all players $i \in S$ satisfy $\tilde{\tau}_i = \sigma_i$. Then, by Lemma~\ref{lemma:bets-basic}, none of the players in $S$ can increase their utility by placing a bet. This would imply that their expected utilities are at most the ones they would get when playing $\vec{\sigma}_S$, which contradicts the (a1) assumption.

\end{proof}

Proposition~\ref{prop:different-strat} implies that if a coalition $S$ with strategy $\vec{\tau}_S$ satisfies (a1), there exist at least one player $i\in S$ and an action $a_i \in A_i$ such that $i$ is more likely to play $a_i$ in $\vec{\tau}_S$ than in $\vec{\sigma}_S$. By Lemma~\ref{lemma:bets-basic}, betting that $i$ is going to play $a_i$ gives a strictly positive utility to all other players. In particular, if $S$ and $\vec{\tau}_S$ satisfy (b1), all other players in $S$ must place bets on $i$ or other players in $S$ with probability $1$ since there exist at least one bet with strictly positive expected utility. Since there is a finite number of possible bets, we have the following result.

\begin{proposition}\label{prop:expected-loss}
Let $q_{(i,a_i)}^{(j, \tau_j)}$ be the probability that player $j$ bets on $(i,a_i)$ when playing $\tau_j$. Then, if $S$ and $\vec{\tau}_S$ satisfy (a1) and (b1), there exists a player $i \in S$ and an action $a_i \in A_i$ such that $$\sum_{j \in S} q_{(i,a_i)}^{(j, \tau_j)} \ge \frac{1}{|A|_{max}},$$
where $|A|_{max} = \max_i \{|A_i|\}$.
\end{proposition}

\begin{proof}
By proposition~\ref{prop:different-strat}, there exists a player $i \in S$ such that $\tilde{\tau}_i \not = \sigma_i$. Because of (b1), if $i$ is the only player in $S$ such that $\tilde{\tau}_i \not = \sigma_i$, all other players in $S$ must bet on $i$ with probability $1$. Since $|S| \ge 2$, it follows that there exist an action $a \in A_i$ such that at least $\frac{1}{|A_i|}$ players bet on $(i,a_i)$ in expectation.

If there are at least two players $i \in S$ such that $\tilde{\tau}_i \not = \sigma_i$ and $S$ and $\vec{\tau}_S$ satisfy (a1) and (b1), then all players in $S$ must place bets on other players in $S$ with probability $1$. This means that, for a fixed $j \in S$, $$\sum_{i \in S, \ a_i \in A_i} q_{(i,a_i)}^{(j, \tau_j)} = 1.$$
If we add the equation above for all $j \in S$ we get that $$\sum_{i,j \in S, \ a_i \in A_i} q_{(i,a_i)}^{(j, \tau_j)} = |S|.$$
Rearranging the terms of the previous sum, we get $$\sum_{i \in S, \ a_i \in A_i} \left(\sum_{j \in S} q_{(i,a_i)}^{(j, \tau_j)} \right) = |S|. $$
Since the outer sum has at most $|S|\cdot|A|_{max}$ elements, there exists $i \in S$ and $a_i \in A_i$ such that $$\sum_{j \in S} q_{(i,a_i)}^{(j, \tau_j)} \ge \frac{|S|}{|S|\cdot|A|_{max}} = \frac{1}{|A|_{max}},$$ as desired.
\end{proof}

Proposition~\ref{prop:expected-loss} shows that,  if $S$ and $\vec{\tau}_S$ satisfy (a1) and (b1), there exists a player $i \in P$ and an action $a_i \in A_i$ such that at least a constant fraction of the players are betting on it in expectation. We next show that, if $C > \max\{M\cdot |A|_{max}, 1\}$, $i$ can strictly increase its utility by decreasing the probability of playing $a_i$. This would contradict the fact that $S$ and $\vec{\tau}_S$ satisfy (b1) and prove Theorem~\ref{thm:main}.

By Proposition~\ref{prop:different-strat}, the probability $p^{\vec{\tau}}_{(i,a_i)}$ that $i$ plays action $a_i$ with strategy $\tilde{\tau}_i$ must satisfy $p^{\vec{\tau}}_{(i,a_i)} > p^{\vec{\sigma}}_{(i,a_i)}$ since, otherwise, players could get a better utility by betting somewhere else. The existence of $a_i$ implies that there exists another action $a'_i \in A_i$ such that $p^{\vec{\tau}}_{(i,a'_i)} < p^{\vec{\sigma}}_{(i,a'_i)}$ (i.e., an action that is played with less probability in $\tilde{\tau}_i$ than in $\sigma_i$).
Since $S$ and $\vec{\tau}_S$ satisfy (b1), it is guaranteed that no player will ever bet on $(i, a_i')$ since, by Lemma~\ref{lemma:bets-basic}, betting on it gives a strictly negative expected utility. Consider a strategy $\tau_i'$ in which $i$ plays in the same way as in $\tau_i$, except that $i$ plays $a_i'$ whenever it would play $a_i$ in $\tau_i$. 

Then, we have that $$
\begin{array}{lll}
u_i^{\vec{\sigma}}(\vec{\sigma}_{-S}, \vec{\tau}_{S \setminus \{i\}}, \tau'_i) - u_i^{\vec{\sigma}}(\vec{\sigma}_{-S}, \vec{\tau}_{S})&
 =  & u_i(\vec{\sigma}_{-S}, \tilde{\tau}_{S \setminus \{i\}}, \tilde{\tau}'_i) - u_i(\vec{\sigma}_{-S}, \tilde{\tau}_{S})\\
&  + & w_i^{\vec{\sigma}}(\vec{\sigma}_{-S}, \vec{\tau}_{S \setminus \{i\}}, \tau'_i) - w_i^{\vec{\sigma}}(\vec{\sigma}_{-S}, \vec{\tau}_{S})\\
 & + &  \ell_i^{\vec{\sigma}}(\vec{\sigma}_{-S}, \vec{\tau}_{S \setminus \{i\}}, \tau'_i) - \ell_i^{\vec{\sigma}}(\vec{\sigma}_{-S}, \vec{\tau}_{S}).
\end{array}
$$

Let $p_{(i, a_i)}^{\tilde{\tau}_i}$ be the probability that $i$ plays action $a_i$ when following strategy $\tilde{\tau}_i$. Then, we can bound the first term in the sum by $$u_i(\vec{\sigma}_{-S}, \tilde{\tau}_{S \setminus \{i\}}, \tilde{\tau}'_i) - u_i(\vec{\sigma}_{-S}, \tilde{\tau}_{S}) \ge -p_{(i, a_i)}^{\tilde{\tau}_i}M$$

since the difference in utility by switching from $a_i$ to $a'_i$ can be at most $-M$. Since $\tau'_i$ bets exactly in the same way as $\tau_i$ does, we have that $w_i^{\vec{\sigma}}(\vec{\sigma}_{-S}, \vec{\tau}_{S \setminus \{i\}}, \tau'_i) - w_i^{\vec{\sigma}}(\vec{\sigma}_{-S}, \vec{\tau}_{S}) = 0$. Finally, since the probability that another player bets on $(i,a_i)$ is at least $\frac{1}{|A_i|}$, the following inequality holds: 
$$\ell_i^{\vec{\sigma}}(\vec{\sigma}_{-S}, \vec{\tau}_{S \setminus \{i\}}, \tau'_i) - \ell_i^{\vec{\sigma}}(\vec{\sigma}_{-S}, \vec{\tau}_{S}) \ge p_{(i, a_i)}^{\tilde{\tau}_i}\frac{C}{|A_i|}.$$

In the equation above, we are using the fact that the only difference between $\tau_i$ and $\tau'_i$ is that $i$ is switching from $a_i$ to another action $a'_i$ that no player ever bets on. If $C > M\cdot |A|_{max}$, then $$u_i^{\vec{\sigma}}(\vec{\sigma}_{-S}, \vec{\tau}_{S \setminus \{i\}}, \tau'_i) - u_i^{\vec{\sigma}}(\vec{\sigma}_{-S}, \vec{\tau}_{S}) \ge p_{(i, a_i)}^{\tilde{\tau}_i}\frac{C}{|A_i|}  - p_{(i, a_i)}^{\tilde{\tau}_i}M > 0.$$

This would imply that $i$ can increase its utility by switching from $\tau_i$ to $\tau'_i$, which contradicts our original assumption that $S$ and $\vec{\tau}_S$ satisfy (b1).

\subsubsection{Part (b)}\label{sec:normal-form-part-b}

The proof of part (b) follows from the following proposition.

\begin{proposition}\label{prop:eq-Nash}
If $\vec{\tau}$ is a Nash equilibrium of $\Gamma^{\sigma}$, then $\tilde{\tau}_i = \sigma_i$ for all $i \in P$.
\end{proposition}

Intuitively, this proposition says that all Nash equilibria of $\Gamma^{\vec{\sigma}}$ are essentially $\vec{\sigma}$ with some bets on top of it. Note that if $\tilde{\tau}_i = \sigma_i$ for all $i \in P$, $w_i$ gives exactly $0$ utility for all strategies of $i$. Since $\ell_i$ is a negative function, Proposition~\ref{prop:eq-Nash} implies that all Nash equilibria of $\Gamma^{\vec{\sigma}}$ are Pareto-dominated by $\vec{\sigma}$, as desired.

\begin{proof}[Proof of Proposition~\ref{prop:eq-Nash}]
Suppose that there exists a player $i$ such that $\tilde{\tau}_i \not= \sigma_i$. An analogous reasoning to the one used for proving Proposition~\ref{prop:expected-loss} shows that there exists a player $i \in P$ and an action $a_i \in A_i$ such that $$\sum_{j \in S}q_{(i,a_i)}^{\tau_j} \ge \frac{1}{|A|_{max}}.$$
Then, analogously to the proof of part (a), it can be shown that $i$ can strictly increase its utility by switching from $a_i$ to another action $a'_i$, contradicting the fact that $\vec{\tau}$ is a Nash equilibrium of $\Gamma^{\sigma}$.
\end{proof}

\subsection {Generalization to Finite Bayesian Games with Independent Types}\label{sec:bayesian-construction}

As shown in Section~\ref{sec:sketch}, the construction of $\Gamma^{\vec{\sigma}}$ used for normal-form games does not work in general for Bayesian games. However, if types are independent, we have the following proposition that states that, if a coalition can defect in such a way that they all increase their expected utility with respect to $\vec{\sigma}$, each player in the coalition must be able to tell the strategy of another member $j$ of the coalition from $\sigma_j$.

\commentout{
there is one subtle difference between the two scenarios. In normal-form games, whenever a player $i$ defects from $\sigma_i$ to $\tau_i$, by definition it is guaranteed that there will always exist an action $a \in A$ such that $a$ is played with more probability in $\tau_i$ than in $\sigma_i$. However, this may not be the case in Bayesian games. For example, consider a game such that $T_i = A_i = \{0,1\}$ for all players $i$, and all types are sampled independently and uniformly from $\{0,1\}$. If $\sigma_i$ is the strategy in which player $i$ plays whatever its type is and $\tau_i$ is the strategy in which it plays the opposite of its type, it is indistinguishable to the other players if $i$ is playing $\sigma_i$ or if $i$ is playing $\tau_i$. In particular, if we construct $\Gamma^{\vec{\sigma}}$ as in Section~\ref{sec:construction}, even if $i$ switches from $\sigma_i$ to $\tau_i$, there would be no bet that gives strictly positive gains to any of the other players.

The key to circumvent the issue above is the fact that, since $\vec{\sigma}$ is a Nash equilibrium of $\Gamma$, 
no player can increase its utility by defecting. This means that, if types are independent and a coalition $S$ defects in such a way that all of its members benefit from colluding, each member $i \in S$ must be able to distinguish the strategies $\tau_j$ played by the other members of the coalition from their original strategies $\sigma_j$. This is made precise in the following proposition.
}

\begin{proposition}\label{prop:distinguishable1}
    Let $\Gamma = (P,T,q,A,U)$ be a Bayesian game with independent types. Let $S$ be a subset of players and $\vec{\tau}_S$ be a strategy profile for players in $S$ such that $u_i(\vec{\sigma}_{-S}, \vec{\tau}_S) > u_i(\vec{\sigma})$ for all $i \in S$. Then, there exists a player $i \in S$ and an action $a_i \in A_i$ such that $$p_{(i,a_i)}^{\tau_i} > p_{(i, a_i)}^{\sigma_i},$$
    where $p_{(i,a_i)}^{\tau_i}$ (resp., $p_{(i,a_i)}^{\sigma_i}$) is the probability that player $i$ plays action $a_i$ when following strategy $\tau_i$ (resp., $\sigma_i$).
\end{proposition}

\begin{proof}
    Suppose that $p_{(i,a_i)}^{\tau_i} = p_{(i, a_i)}^{\sigma_i}$ for all $i \in S$ and all $a_i \in A_i$. Then, since types are independent, if we fix a player $i \in S$ and a type $t_i$ of $i$, the distribution over action profiles $\vec{a}_{-i}$ played by the other players $j$ conditioned to the fact that $i$ has type $t_i$ is identical when players in $S$ play $\vec{\tau}_S$ and when players in $S$ play $\vec{\sigma}_S$. Therefore, we would have that $$u_i(\vec{\sigma}_{-S}, \vec{\tau}_S) = u_i(\vec{\sigma}_{-i}, \tau_i)$$ for all $i \in S$, and thus $$ u_i(\vec{\sigma}_{-i}, \tau_i) >  u_i(\vec{\sigma}),$$ which contradicts the fact that $\vec{\sigma}$ is a Nash equilibrium of $\Gamma$.
\end{proof}

Note that the proof of Proposition~\ref{prop:distinguishable1} requires that the types of the players are independent. In fact, if they were not, Example~\ref{example:indistinguishable} shows that players can defect in such a way that they all benefit while no player $i$ can tell between the strategy of another player $j$ and $\sigma_j$. However, in the case of independent types, Proposition~\ref{prop:distinguishable1} shows that we can use the same approach as in the case of normal-form games. Define $A^{\vec{\sigma}}$ as in Section~\ref{sec:construction}, and let $$\begin{array}{lll}
 u^{\vec{\sigma}}_i(t_i, ((a_1, b_1), (a_2, b_2), \ldots, (a_n, b_n))) &  :=  & u_i(t_i, (a_1, a_2, \ldots, a_n)) \\
  & + & w_i^{\vec{\sigma}}((a_1, b_1), (a_2, b_2), \ldots, (a_n, b_n))\\
  & + & \ell_i^{\vec{\sigma}}((a_1, b_1), (a_2, b_2), \ldots, (a_n, b_n)),
  \end{array}
  $$

where $w_i$ and $\ell_i$ are defined in the same way as in Section~\ref{sec:construction} and $C > \max\{M \cdot |A|_{\max}, 1\}$, as in Section~\ref{sec:proof}. An analogous reasoning to that of the proof of Proposition~\ref{prop:expected-loss} shows that, if $S$ and $\vec{\tau}_S$ satisfy that $u_i^{\vec{\sigma}}(\vec{\sigma}_{-S}, \vec{\tau}_S) > u_i^{\vec{\sigma}}(\vec{\sigma})$ for all $i \in S$, and that $u_i^{\vec{\sigma}}(\vec{\sigma}_{-S}, \vec{\tau}_S) \ge u_i^{\vec{\sigma}}(\vec{\sigma}_{-S}, \vec{\tau}_{S \setminus \{i\}}, \tau'_i)$ for all $i \in S$ and all strategies $\tau'_i$ for $i$, then there exists a player $i$ and an action $a_i$ such that, regardless of $i$'s type, at least $\frac{1}{|A|_{max}}$ players bet on $(i, a_i)$ in expectation. The rest of the proof proceeds as in the case of normal-form games.

\section{Proof of Theorem~\ref{thm:main2} for Finite Bayesian Games}\label{sec:arbitrary-bayesian}

The analysis of Example~\ref{example:indistinguishable} given in Section~\ref{sec:sketch} shows that the construction used for normal-form games and for Bayesian games with independent types does not quite work in general for Bayesian games. However, note that if all three players play $\vec{\tau}$ in Example~\ref{example:indistinguishable}, even though player $1$ cannot tell if player $2$ or player $3$ are going to play $0$ or $1$, she knows the following. If her type is $0$, then players $2$ and $3$ are going to play either $(0,1)$ or $(1,0)$. Otherwise, they are going to play $(0,0)$ or $(1,1)$. This means that, even if $i$ cannot guess the action of another player $j$ with higher probability in $\tau_j$ than in $\sigma_j$, she can guess the action profile of all other players with higher probability. This suggests that, instead of allowing players to bet on which action another player will play, they should be able to bet on action profiles of other players. This motivates the following construction. Define $\Gamma^{\vec{\sigma}} := (P, T, q, A^{\vec{\sigma}}, U^{\vec{\sigma}})$. We define the action set $A_i^{\vec{\sigma}}$ of each player as $A^{\vec{\sigma}}_i := A_i \cup A$. 
As in Section~\ref{sec:construction}, we can view each element of $A^{\vec{\sigma}}_i$ with a pair $(a,b)$, where $a$ is the actual action that $i$ plays and $b \in A_{-i}$ is $i$'s bet over the action profile of other players. Again, for convenience, whenever $i$ plays an action $a_i$ without placing a bet, we will use the pair $(a_i, \bot)$ instead of $a_i$. The utility of each player $i$ is defined as follows:

$$\begin{array}{lll}
 u^{\vec{\sigma}}_i((a_1, b_1), (a_2, b_2), \ldots, (a_n, b_n)) &  :=  & u_i(a_1, a_2, \ldots, a_n) \\
  & + & w_i^{\vec{\sigma}}((a_1, b_1), (a_2, b_2), \ldots, (a_n, b_n))\\
  & + & \ell_i^{\vec{\sigma}}((a_1, b_1), (a_2, b_2), \ldots, (a_n, b_n)),
  \end{array}
  $$
where $w_i^{\vec{\sigma}}$ and $\ell_i^{\vec{\sigma}}$ are defined similarly to their counterparts in Section~\ref{sec:construction}. In this case, let $p^{\vec{\sigma}}_{i, t_i, \vec{a}_{-i}}$ denote the probability that the action profile played by $P_{-i}$ is $\vec{a}_{-i}$, conditioned to the fact that all players play $\vec{\sigma}$ and that player $i$ has type $t_i$. Then, we can define $w_i^{\vec{\sigma}}$ as follows:
$$w_i^{\vec{\sigma}}((a_1, b_1), (a_2, b_2), \ldots, (a_n, b_n)) := 
\left\{
\begin{array}{ll}
0 & \mbox{if } b_i = \bot\\
1 - p_{(i, t_i, \vec{a}_{-i})}^{\vec{\sigma}} & \mbox{if } b_i \mbox{ is of the form } \vec{a}_{-i}\\
- p_{(i, t_i, \vec{a}_{-i})}^{\vec{\sigma}} & \mbox{if } b_i \mbox{ is of the form } \vec{a}_i' \mbox{ with } \vec{a}_i' \not= \vec{a}_{-i}.\\
\end{array}
\right.$$

Finally, we define $\ell_i^{\vec{\sigma}} := \sum_{j \not = i} \ell_{(i, j)}^{\vec{\sigma}}$, where $$
  \ell_{(i,j)}^{\vec{\sigma}}((a_1, b_1), (a_2, b_2), \ldots, (a_n, b_n)) := 
\left\{
\begin{array}{ll}
0 & \mbox{if } b_j = \bot \mbox{ or } b_j \mbox{ is of the form } \vec{a}_{-j}' \mbox{ with } \vec{a}_{-j}' \not= \vec{a}_{-j}\\
-C & \mbox{otherwise,}
\end{array}\right.$$
and $C = 2|A|\cdot (M+1)$. Note that, in this case, $\ell_{(i,j)}$ does not depend on $i$.

To prove that this construction satisfies the properties of Theorem~\ref{thm:main} we need a few preliminary lemmas. Before stating the first of these, we need the following definition. Given a player $i \in P$, a type $t_i \in T_i$, and a strategy profile $\vec{\tau}$, define $\Delta_{(i,t_i)}^{\vec{\tau}}$ as $$\Delta_{(i,t_i)}^{\vec{\tau}} = \sum_{\vec{a}_{-i} \in A_{-i}} \left|p_{(i, t_i, \vec{a}_{-i})}^{\vec{\sigma}} - p_{(i, t_i, \vec{a}_{-i})}^{\vec{\tau}}\right|.$$
Intuitively, $\Delta_{(i,t_i)}^{\vec{\tau}}$ can be seen as the ``Manhattan distance'' between the distributions over actions played by $P_{-i}$ that are generated when players play $\vec{\sigma}$ and $\vec{\tau}$ respectively, conditioned to the fact that $i$ has type $t_i$. Our first lemma states that, if players play $\vec{\tau}$, the expected utility of a player $i$ with type $t_i$ is bounded by its expected utility when players play $\vec{\sigma}$ plus $\Delta_{(i,t_i)}^{\vec{\tau}} \cdot M$.

\begin{lemma}\label{lemma:diff-utility-bayesian}
Given a player $i \in P$, its type $t_i \in T_i$, and a strategy profile $\vec{\tau}$ for $\Gamma$, denote by $u_{i}(t_i, \vec{\tau})$ the expected utility of $i$ when players play $\vec{\tau}$, conditioned to the fact that $i$ has type $t_i$. Then, $$u_{i}(t_i, \vec{\tau}) \le u_{i}(t_i, \vec{\sigma}) + \Delta_{(i,t_i)}^{\vec{\tau}} \cdot M.$$
\end{lemma}

\begin{proof}
Player $i$'s expected utility when players play $\vec{\tau}$ conditioned to the fact that $i$ has type $t_i$ is given by $$u_{i}(t_i, \vec{\tau}) = \sum_{a_i \in A_i} \tau_i(t_i,a_i) \left(\sum_{a_{-i} \in A_{-i}} p_{(i, t_i, \vec{a}_{-i})}^{\vec{\tau}}u_i(t_i,\vec{a})\right),$$
where $\tau_i(t_i, a_i)$ denotes the probability that $i$ plays $a_i$ when following strategy $\tau_i$, conditioned to the fact that it has type $t_i$. 

We can rewrite the sum as 

$$u_{i}(t_i, \vec{\tau}) = \sum_{a_i \in A_i} \tau_i(t_i,a_i) \left(\sum_{a_{-i} \in A_{-i}} p_{(i, t_i, \vec{a}_{-i})}^{\vec{\sigma}}u_i(t_i,\vec{a}) + (p_{(i, t_i, \vec{a}_{-i})}^{\vec{\tau}} - p_{(i, t_i, \vec{a}_{-i})}^{\vec{\sigma}})u_i(t_i,\vec{a})\right).$$

Note that, since $\vec{\sigma}$ is a Nash equilibrium, we have that $$\sum_{a_i \in A_i} \tau_i(t_i,a_i) \left(\sum_{a_{-i} \in A_{-i}} p_{(i, t_i, \vec{a}_{-i})}^{\vec{\sigma}}u_i(t_i,\vec{a})\right) \le u_i(t_i, \vec{\sigma}).$$

It remains to show that

$$\sum_{a_i \in A_i} \tau_i(t_i,a_i) \left(\sum_{a_{-i} \in A_{-i}} (p_{(i, t_i, \vec{a}_{-i})}^{\vec{\tau}} - p_{(i, t_i, \vec{a}_{-i})}^{\vec{\sigma}})u_i(t_i,\vec{a})\right) \le \Delta_{(i,t_i)}^{\vec{\tau}} \cdot M.$$

Let $M_{min}$ be the minimum utility of any player in $\Gamma$ and $M_{max}$ be the maximum utility of any player in $\Gamma$ (recall that $M = M_{max} - M_{min}$). Then, $$\sum_{a_{-i} \in A_{-i}} (p_{(i, t_i, \vec{a}_{-i})}^{\vec{\tau}} - p_{(i, t_i, \vec{a}_{-i})}^{\vec{\sigma}})u_i(t_i,\vec{a}) = \sum_{a_{-i} \in A_{-i}} (p_{(i, t_i, \vec{a}_{-i})}^{\vec{\tau}} - p_{(i, t_i, \vec{a}_{-i})}^{\vec{\sigma}})(M_{min} + (u_i(t_i,\vec{a}) - M_{min})).$$

Since $\sum_{a_{-i} \in A_{-i}} p_{(i, t_i, \vec{a}_{-i})}^{\vec{\tau}} = \sum_{a_{-i} \in A_{-i}} p_{(i, t_i, \vec{a}_{-i})}^{\vec{\sigma}} = 1$, we have that $$
\begin{array}{ll}
& \sum_{a_{-i} \in A_{-i}} (p_{(i, t_i, \vec{a}_{-i})}^{\vec{\tau}} - p_{(i, t_i, \vec{a}_{-i})}^{\vec{\sigma}})(M_{min} + (u_i(t_i,\vec{a}) - M_{min}))\\
= & \sum_{a_{-i} \in A_{-i}} (p_{(i, t_i, \vec{a}_{-i})}^{\vec{\tau}} - p_{(i, t_i, \vec{a}_{-i})}^{\vec{\sigma}})(u_i(t_i,\vec{a}) - M_{min}) \\
\le & \sum_{a_{-i} \in A_{-i}} \left|p_{(i, t_i, \vec{a}_{-i})}^{\vec{\tau}} - p_{(i, t_i, \vec{a}_{-i})}^{\vec{\sigma}}\right| \cdot (M_{max} - M_{min}) \\
= & \Delta_{(i,t_i)}^{\vec{\tau}} \cdot M.
\end{array}
$$

Therefore, 

$$
\begin{array}{ll}
 & \sum_{a_i \in A_i} \tau_i(t_i,a_i) \left(\sum_{a_{-i} \in A_{-i}} (p_{(i, t_i, \vec{a}_{-i})}^{\vec{\tau}} - p_{(i, t_i, \vec{a}_{-i})}^{\vec{\sigma}})u_i(t_i,\vec{a})\right) \\
\le & (\Delta_{(i,t_i)}^{\vec{\tau}} \cdot M) \sum_{a_i \in A_i} \tau_i(t_i,a_i)  \\
= & \Delta_{(i,t_i)}^{\vec{\tau}} \cdot M,
\end{array}
$$

as desired.

\end{proof}

The second preliminary lemma gives the expected utility that a player can get by betting on the other players' actions.

\begin{lemma}\label{lem:bayesian-basic-bets}
    Given a player $i \in P$, its type $t_i \in T_i$, and a strategy profile $\vec{\tau}$ for $\Gamma^{\vec{\sigma}}$. The expected utility that $i$ gets when betting on action profile $\vec{a}_{-i}$ is given  by $p_{(i, t_i, \vec{a}_{-i})}^{\vec{\tau}} - p_{(i, t_i, \vec{a}_{-i})}.$
\end{lemma}

The proof of Lemma~\ref{lem:bayesian-basic-bets} is equivalent to that of Lemma~\ref{lemma:bets-basic}. The following corollary is a direct implication.

\begin{corollary}\label{col:max-wins-bets}
Given a player $i \in P$, its type $t_i \in T_i$, and a Nash equilibrium $\vec{\tau}$ for $\Gamma^{\vec{\sigma}}$. The maximum expected utility $w_i (t_i, \vec{\tau}_i)$ that $i$ can get by betting when players play $\vec{\tau}$ and $i$ has type $t_i$ satisfies $$w_i (t_i, \vec{\tau}_i) \le \Delta_{(i, t_i)}^{\vec{\tau}}$$
\end{corollary}

\begin{proof}
By Lemma~\ref{lem:bayesian-basic-bets}, the maximum utility that $i$ can get by betting, conditioned to the fact that $i$ has type $t_i$ and players play $\vec{\tau}$ is given by $$\max_{\vec{a}_{-i} \in A_{-i}}\{p_{(i, t_i, \vec{a}_{-i})}^{\vec{\tau}} - p_{(i, t_i, \vec{a}_{-i})}^{\vec{\sigma}}\}.$$

By definition, this value is at most $\Delta_{(i, t_i)}^{\vec{\tau}}$.
\end{proof}

The third and final preliminary lemma gives a lower bound on the expected losses of other players incurred by $i$'s bets.

\begin{lemma}\label{lem:bayesian-expected-loss}
Given a player $i \in P$, its type $t_i \in T_i$, and a Nash equilibrium $\vec{\tau}$ for $\Gamma^{\vec{\sigma}}$, denote by $\ell_{j,i}^{\vec{\sigma}}(t_i,\vec{\tau})$ the expected utility that $j$ loses because of $i$'s bet, conditioned to the fact that $i$ plays $\tau_i$ with type $t_i$. Then, $$\ell_{j,i}^{\vec{\sigma}}(t_i,\vec{\tau}) \le - \frac{\Delta_{(i, t_i)}^{\vec{\tau}} C}{2|A_{-i}|}.$$
\end{lemma}

\begin{proof}
If $\Delta_{(i, t_i)}^{\vec{\tau}} = 0$, the lemma is automatically true since $\ell_{j,i}^{\vec{\sigma}}(t_i,\vec{\tau}) \le 0$. Therefore, for the rest of the proof we will assume that $\Delta_{(i, t_i)}^{\vec{\tau}} > 0$, and thus that there exists an action profile $\vec{a}_{-i}$ such that $p_{(i, t_i, \vec{a}_{-i})}^{\vec{\tau}} - p_{(i, t_i, \vec{a}_{-i})}^{\vec{\sigma}} > 0$. By Lemma~\ref{lem:bayesian-basic-bets}, this implies that there exists at least one bet for $i$ that gives her strictly positive utility, and therefore that $i$ will always bet with probability $1$ since $\vec{\tau}$ is a Nash equilibrium. 

Let $$X := \sum_{\substack{\vec{a}_{-i} \in A_{-i} \\ p_{(i, t_i, \vec{a}_{-i})}^{\vec{\tau}} - p_{(i, t_i, \vec{a}_{-i})} \ge 0}} (p_{(i, t_i, \vec{a}_{-i})}^{\vec{\tau}} - p_{(i, t_i, \vec{a}_{-i})}^{\vec{\sigma}})$$ and $$Y := \sum_{\substack{\vec{a}_{-i} \in A_{-i} \\ p_{(i, t_i, \vec{a}_{-i})}^{\vec{\tau}} - p_{(i, t_i, \vec{a}_{-i})} < 0}} (p_{(i, t_i, \vec{a}_{-i})}^{\vec{\tau}} - p_{(i, t_i, \vec{a}_{-i})}^{\vec{\sigma}}).$$

We have that $$X + Y = \sum_{\vec{a}_{-i}\in A_{-i}} p_{(i, t_i, \vec{a}_{-i})}^{\vec{\tau}} - \sum_{\vec{a}_{-i}\in A_{-i}} p_{(i, t_i, \vec{a}_{-i})}^{\vec{\sigma}} = 0.$$

Moreover, by definition, $X - Y = \Delta_{(i, t_i)}^{\vec{\tau}}.$ Putting the last two equations together we get that $X = \frac{\Delta_{(i, t_i)}^{\vec{\tau}}}{2}$. Therefore, since $X$ is a sum of at most $|A_{-i}|$ terms, there exists an action profile $\vec{a}_{-i} \in A_{-i}$ such that 
$$ p_{(i, t_i, \vec{a}_{-i})}^{\vec{\tau}} - p_{(i, t_i, \vec{a}_{-i})}^{\vec{\sigma}} \ge \frac{1}{|A_{-i}|} \cdot \frac{\Delta_{(i, t_i)}^{\vec{\tau}}}{2} = \frac{\Delta_{(i, t_i)}^{\vec{\tau}}}{2|A_{-i}|}.$$

We show next that the loss of every other player when $i$ bets on an action that maximizes $p_{(i, t_i, \vec{a}_{-i})}^{\vec{\tau}} - p_{(i, t_i, \vec{a}_{-i})}^{\vec{\sigma}}$ is upper-bounded by $- \frac{\Delta_{(i, t_i)}^{\vec{\tau}} C}{2|A_{-i}|}$. Note that this completes the proof since, by Lemma~\ref{lem:bayesian-basic-bets}, $i$ will always bet on actions that maximize $p_{(i, t_i, \vec{a}_{-i})}^{\vec{\tau}} - p_{(i, t_i, \vec{a}_{-i})}^{\vec{\sigma}}$ since $\vec{\tau}$ is a Nash equilibrium.

Let $\vec{a}_{-i}$ be an action profile that maximizes $p_{(i, t_i, \vec{a}_{-i})}^{\vec{\tau}} - p_{(i, t_i, \vec{a}_{-i})}^{\vec{\sigma}}$. The expected loss $\ell_{j,i}^{\vec{\sigma}}(t_i,\vec{a}_{-i})$ of every other player $j$ when $i$ bets on $\vec{a}_{-i}$  and $i$ has type $t_i$ is given by $$\ell_{j,i}^{\vec{\sigma}}(t_i,\vec{a}_{-i}) = p_{(i, t_i, \vec{a}_{-i})}^{\vec{\tau}} \cdot (-C).$$

Since $p_{(i, t_i, \vec{a}_{-i})}^{\vec{\tau}} \ge  p_{(i, t_i, \vec{a}_{-i})}^{\vec{\sigma}} + \frac{\Delta_{(i, t_i)}^{\vec{\tau}}}{2|A_{-i}|} \ge \frac{\Delta_{(i, t_i)}^{\vec{\tau}}}{2|A_{-i}|}$, we get that $$\ell_{j,i}^{\vec{\sigma}}(t_i,\vec{a}_{-i}) \le   - \frac{\Delta_{(i, t_i)}^{\vec{\tau}} C}{2|A_{-i}|},$$

and this completes the proof.
\end{proof}

With Lemmas~\ref{lemma:diff-utility-bayesian} and \ref{lem:bayesian-expected-loss}, and Corollary~\ref{col:max-wins-bets}, we are ready to prove Theorem~\ref{thm:main}. Let $S$ be a coalition of players and $\vec{\tau}_S$ be a strategy profile for players in $S$ such that (a1) $u_i^{\vec{\sigma}}(\vec{\sigma}_{-S}, \vec{\tau}_S) > u_i^{\vec{\sigma}}(\vec{\sigma})$ for all $i \in S$, and (b1) $u_i^{\vec{\sigma}}(\vec{\sigma}_{-S}, \vec{\tau}_S) \ge u_i^{\vec{\sigma}}(\vec{\sigma}_{-S}, \vec{\tau}_{S \setminus \{i\}}, \tau'_i)$ for all $i \in S$ and all strategies $\tau'_i$ for $i$. By construction, we have that $$
\begin{array}{lll}
u_i^{\vec{\sigma}}(\vec{\sigma}_{-S}, \vec{\tau}_S) & = & u_i(\vec{\sigma}_{-S}, \tilde{\tau}_S)\\
& + & w_i^{\vec{\sigma}}(\vec{\sigma}_{-S}, \vec{\tau}_S)\\
& + & \ell_i^{\vec{\sigma}}(\vec{\sigma}_{-S}, \vec{\tau}_S).
\end{array}
$$

If we denote by $p(i, t_i)$ the probability that $i$ gets type $t_i$ according to distribution $p$, can rewrite the previous identity as  

$$
\begin{array}{lll}
u_i^{\vec{\sigma}}(\vec{\sigma}_{-S}, \vec{\tau}_S) & = & \sum_{t_i \in T_i} p(i, t_i) u_i(t_i, (\vec{\sigma}_{-S}, \tilde{\tau}_S))\\
& + & \sum_{t_i \in T_i} p(i, t_i) w_i^{\vec{\sigma}}(t_i, (\vec{\sigma}_{-S}, \vec{\tau}_S))\\
& + & \sum_{j \in S \setminus \{i\}} \ell_{(i,j)}^{\vec{\sigma}}(\vec{\sigma}_{-S}, \vec{\tau}_S).
\end{array}
$$

Applying Lemma~\ref{lemma:diff-utility-bayesian} and Corollary~\ref{col:max-wins-bets} we get that 
$$
\begin{array}{lll}
u_i^{\vec{\sigma}}(\vec{\sigma}_{-S}, \vec{\tau}_S) & \le & \sum_{t_i \in T_i} p(i, t_i) (u_i(t_i, \vec{\sigma}) + \Delta^{(\vec{\sigma}_{-S}, \tilde{\tau}_S)}_{(i, t_i)} \cdot M)\\
& + & \sum_{t_i \in T_i} p(i, t_i) \Delta^{(\vec{\sigma}_{-S}, \vec{\tau}_S)}_{(i, t_i)}\\
& + & \sum_{j \in S \setminus \{i\}} \ell_{(i,j)}^{\vec{\sigma}}(\vec{\sigma}_{-S}, \vec{\tau}_S),
\end{array}
$$

which can be written as 

$$
\begin{array}{lll}
u_i^{\vec{\sigma}}(\vec{\sigma}_{-S}, \vec{\tau}_S) & \le & \sum_{t_i \in T_i} p(i, t_i) (u_i(t_i, \vec{\sigma}) + \Delta^{(\vec{\sigma}_{-S}, \vec{\tau}_S)}_{(i, t_i)} \cdot (M+1))\\
& + & \sum_{j \in S \setminus \{i\}} \ell_{(i,j)}^{\vec{\sigma}}(\vec{\sigma}_{-S}, \vec{\tau}_S),
\end{array}
$$

By linearity of expectation, this implies that 
\begin{equation}\label{eq:expectation-bayesian}
\begin{array}{lll}
\sum_{i \in S} u_i^{\vec{\sigma}}(\vec{\sigma}_{-S}, \vec{\tau}_S) & \le & \sum_{i \in S} \sum_{t_i \in T_i} p(i, t_i) (u_i(t_i, \vec{\sigma}) + \Delta^{(\vec{\sigma}_{-S}, \vec{\tau}_S)}_{(i, t_i)} \cdot (M+1))\\
& + & \sum_{i \in S} \sum_{j \in S \setminus \{i\}} \ell_{(i,j)}^{\vec{\sigma}}(\vec{\sigma}_{-S}, \vec{\tau}_S).
\end{array}
\end{equation}

Note that $\sum_{i \in S} \sum_{j \in S \setminus \{i\}} \ell_{(i,j)}^{\vec{\sigma}}(\vec{\sigma}_{-S}, \vec{\tau}_S) = \sum_{i \in S} \sum_{j \in S \setminus \{i\}} \ell_{(j,i)}^{\vec{\sigma}}(\vec{\sigma}_{-S}, \vec{\tau}_S)$. Moreover, we can write 
$\ell_{(j,i)}^{\vec{\sigma}}(\vec{\sigma}_{-S}, \vec{\tau}_S)$ as $\sum_{t_i \in T_i} p(i, t_i) \ell_{(j,i)}^{\vec{\sigma}}(t_i, (\vec{\sigma}_{-S}, \vec{\tau}_S))$. Therefore, by Lemma~\ref{lem:bayesian-expected-loss}, we get that 

$$
\ell_{(j,i)}^{\vec{\sigma}}(\vec{\sigma}_{-S}, \vec{\tau}_S) \le \sum_{t_i \in T_i} - p(i, t_i)\cdot \frac{ \Delta^{(\vec{\sigma}_{-S}, \vec{\tau}_S)}_{(i, t_i)} C}{2|A_{-i}|}.
$$

This implies that $$
\sum_{i \in S} \sum_{j \in S \setminus \{i\}} \ell_{(j,i)}^{\vec{\sigma}}(\vec{\sigma}_{-S}, \vec{\tau}_S) \le \sum_{i \in S}\sum_{t_i \in T_i} -p(i, t_i) \cdot (|S| - 1) \cdot 
\frac{ \Delta^{(\vec{\sigma}_{-S}, \vec{\tau}_S)}_{(i, t_i)} C}{2|A_{-i}|}$$

Inserting this into Equation~\ref{eq:expectation-bayesian} gives that $$\sum_{i \in S} u_i^{\vec{\sigma}}(\vec{\sigma}_{-S}, \vec{\tau}_S) \le \sum_{i \in S} \sum_{t_i \in T_i} p(i, t_i) \left[(u_i(t_i, \vec{\sigma}) + \Delta^{(\vec{\sigma}_{-S}, \vec{\tau}_S)}_{(i, t_i)} \cdot (M+1)) - (|S| - 1) \cdot 
\frac{ \Delta^{(\vec{\sigma}_{-S}, \vec{\tau}_S)}_{(i, t_i)} C}{2|A_{-i}|} \right].$$

Since $C = 2|A|\cdot (M+1)$, it follows that $$(|S| - 1) \cdot 
\frac{ \Delta^{(\vec{\sigma}_{-S}, \vec{\tau}_S)}_{(i, t_i)} C}{2|A_{-i}|} \ge \Delta^{(\vec{\sigma}_{-S}, \vec{\tau}_S)}_{(i, t_i)} \cdot (M+1),$$

and therefore 

$$\sum_{i \in S} u_i^{\vec{\sigma}}(\vec{\sigma}_{-S}, \vec{\tau}_S) \le \sum_{i \in S} \sum_{t_i \in T_i} p(i, t_i) u_i(t_i, \vec{\sigma}) = \sum_{i\in S} u_i(\vec{\sigma}).$$

This implies that there exists at least one player $i \in S$ such that $u_i^{\vec{\sigma}}(\vec{\sigma}_{-S}, \vec{\tau}_S) \le u_i(\vec{\sigma})$, which contradicts the (a1) assumption. This proves Theorem~\ref{thm:main} part (a). Part (b) follows from the same argument with $S = P$.

\section{Proof of Theorems~\ref{thm:main} and \ref{thm:main2}, Infinite Case}\label{sec:infinite}

As stated in Section~\ref{sec:results}, the proof of Theorem~\ref{thm:main} can be generalized to games in which the sets of actions and types are infinite. We'll prove this generalization in the case where $\Gamma$ is a normal-form game. The generalization in the case of Bayesian games is identical, except that we use the proofs and constructions from Sections~\ref{sec:bayesian-construction} and \ref{sec:arbitrary-bayesian} as a starting point (as opposed to the proofs and constructions of Section~\ref{sec:proof}).

For ease of exposition, we will begin by generalizing the constructions given in previous sections to games with a countable number of actions. The proof for arbitrary sets is given in Section~\ref{sec:uncountable} and follows the same lines.

\subsection{Countable Case}\label{sec:countable}

As suggested by Example~\ref{example:infinite}, the proof given in Section~\ref{sec:construction} does not quite work when the sets of actions are infinite. In fact, the construction for finite games uses the fact that if a coalition $S$ defects, there exists a player $i$ and an action $a$ such that $\frac{1}{|A|_{max}}$ of the players bet on $(i,a)$ in expectation. This allowed us to adjust $C$ so that it is never beneficial for $i$ to play $a$. If the game is infinite, there is no lower bound of the expected number of players that bet on a particular action, and therefore there is no value of $C$ that does the trick.

We can avoid this problem by slightly modifying the bet space. Suppose that instead of betting on pairs $(i,a)$ with $i \in P$ and $a \in A_i$, players can bet on pairs $(i,X)$ with $i \in P$ and $X \subseteq A_i$. Given the actions and bets $(a_1, b_1), \ldots, (a_n, b_n)$ of the players, we define the utilities of $\Gamma^{\vec{\sigma}}$ as in Section~\ref{sec:construction}, except that 

$$w_i^{\vec{\sigma}}((a_1, b_1), (a_2, b_2), \ldots, (a_n, b_n)) := 
\left\{
\begin{array}{ll}
0 & \mbox{if } b_i = \bot\\
1 - p_{(j, X)}^{\vec{\sigma}} & \mbox{if } b_i \mbox{ is of the form } (j,X) \mbox{ with 
} a_j \in X\\
- p_{(j, X)}^{\vec{\sigma}} & \mbox{if } b_i \mbox{ is of the form } (j,X) \mbox{ with 
} a_j \not\in X\\
\end{array}
\right.$$

where $p_{(j, X)}^{\vec{\sigma}}$ is the probability that the first component of $j$'s action is in $X$, assuming that $j$ uses strategy $\sigma_j$, and

  $$\ell_{(i,j)}^{\vec{\sigma}}((a_1, b_1), (a_2, b_2), \ldots, (a_n, b_n)) := 
\left\{
\begin{array}{ll}
0 & \mbox{if } b_j = \bot \mbox{ or } b_j \mbox{ is of the form } (i,X) \mbox{ with } a \not \in X\\
-M - 1& \mbox{otherwise,}
\end{array}\right.$$

 where $M$ is the difference between the supremum and the infimum utility that a player can get in $\Gamma$. An analogous reasoning to that  of Lemma~\ref{lemma:bets-basic} shows that the expected utility $x^{\tau_j}_{(j, a_j)}$ that a player gets by betting on $(j,X)$ conditioned to the fact that $j$ is playing strategy $\tau_j$, is given by $$x^{\tau_j}_{(j, a_j)} := p_{(j, X)}^{\tau_j} - p_{(j, X)}^{\sigma_j}.$$

This implies that if player $j$ plays some action $a$ with more probability in $\tau_j$ than in $\sigma_j$, including action $a$ in $X$ strictly increases the expected utility of betting on $(j, X)$. This proves the following lemma.

\begin{lemma}\label{lem:same-bets}
Given a normal-form game $\Gamma = (P, A, U)$, let $S \subset P$ be a subset of players and $\vec{\tau}_S$ be a strategy for players in $S$ in $\Gamma^{\vec{\sigma}}$ such that $u_i^{\vec{\sigma}}(\vec{\sigma}_{-S}, \vec{\tau}_{S}) \ge  u_i^{\vec{\sigma}}(\vec{\sigma}_{-S}, \vec{\tau}_{S \setminus\{i\}}, \tau'_i)$ for all $i \in S$ and all strategies $\tau'_i$ for $i$. For each player $i \in S$, let $X_i^>$ (resp., $X_i^<$) be the set of all actions $a \in A_i$ such that the probability that $i$ plays $a$ with $\tilde{\tau}_i$ is strictly greater (resp., smaller) than the probability that $i$ plays $a$ with $\sigma_i$. Then, all bets of the form $(i,X)$ that may occur with positive probability when playing $(\vec{\sigma}_{-S}, \vec{\tau}_{S})$ must satisfy that $X_i^> \subseteq X$ and $X_i^< \cap X = \emptyset$.
\end{lemma}

With Lemma~\ref{lem:same-bets}, we can use a similar argument to the one used in the finite case. Suppose that a coalition $S$ defects from the proposed strategy  following strategy profile $\vec{\tau}_S$ in such a way that (a1) all players in $S$ get more utility with $\vec{\tau}_S$ than with $\vec{\sigma}$ and (b1) no player in $S$ can further increase its utility by defecting from $\vec{\tau}_S$. Then, it is straightforward to check that, if (a1) is satisfied, all players in $S$ bet with probability $1$. Therefore, there exists a player $i \in S$ such that at least one other player bets on $i$ in expectation. By Lemma~\ref{lem:same-bets}, all of these bets must include all elements in $X_i^>$ and none of them includes any element of $X_i^<$. Therefore, if $i$ reduces the probability of playing any action in $X_i^>$ by $\delta$ and increases the probability of playing any action in $X_i^<$ by $\delta$, an analogous reasoning to the one at the end of Section~\ref{sec:proof} shows that the utility of $i$ increases at least by $$\delta (M+1) - \delta M,$$

which contradicts the (b1) assumption. The proof for Bayesian games follow the same generalization (betting on sets of actions or action profiles rather than on individual actions or action profiles), using as a basis the constructions in Sections~\ref{sec:bayesian-construction} and ~\ref{sec:arbitrary-bayesian}.

\subsection{Uncountable Case}\label{sec:uncountable}

The idea used for the countable case can be generalized for any probability space, although the proof is  much more technical. We begin by giving the basic notions on probability theory, and then move to the desired generalization.

\subsubsection{Basic Definitions}

Given a set $X$, a \emph{$\sigma$-algebra} on $X$ is a collection $\Sigma$ of subsets of $X$ that is closed under complement, countable unions, and countable intersections. This means that, 
\begin{itemize}
    \item $Y \in \Sigma \Longrightarrow \overline{Y} \in \Sigma$, where $\overline{Y}$ denotes the complement of $Y$ in $X$ (i.e., $X \setminus Y$).
    \item If $\{X_i\}_{i = 1}^{\infty}$ is a countable collection of elements of $\Sigma$, then both $\bigcup_{i \ge 1} X_i$ and $\bigcup_{i \ge 1} X_i$ are also elements of $\Sigma$.
\end{itemize}

 A pair $(X, \Sigma)$ consisting of a set $X$ and a $\sigma$-algebra $\Sigma$ on $X$ is called a \emph{measurable space}. Given a measurable space $(X, \Sigma)$, a measure on $(X, \Sigma)$ is a function $\mu : \Sigma \longrightarrow \mathbb{R} \cup \{-\infty, +\infty\}$ that satisfies 
\begin{itemize}
    \item $\mu(\emptyset) = 0$.
    \item \textbf{Non-negativity:} $\mu(Y) \ge 0$ for all $Y \in \Sigma$.
    \item \textbf{Countable additivity:} For all countable collections  $\{X_i\}_{i = 1}^{\infty}$ of pairwise disjoint sets in $\Sigma$, $\mu(\bigcup_{i = 1}^{\infty} X_i) = \sum_{i = 1}^{\infty} \mu(X_i)$.
\end{itemize}
Intuitively, a measure is a non-negative function that assigns $0$ to the empty set and such that the measure of a countable union disjoint sets is the sum of their measures.

With these definitions, we can define a probability space. A probability space is a tuple $(X, \Sigma, P)$, where $(X, \Sigma)$ is a $\sigma$-algebra, and $P$ is a measure such that the image of all elements of $\Sigma$ is in $[0,1]$, and such that the measure of the whole space is $1$. More precisely, $P$ must satisfy that

\begin{itemize}
    \item $P(E) \in [0,1]$ for all $E \in \Sigma$.
    \item $P(X) = 1$.
\end{itemize}

In the next sections we will assume that, in a normal-form game $\Gamma$, the action sets $A_i$ of each player have an embedded $\sigma$-algebra $\Sigma_i$, and the strategy of each player $i$ consists of choosing a measure $Q_i$ such that $(A_i, \Sigma_i, Q_i)$ is a probability space. Intuitively, this means that $i$ cannot choose which sets of $A_i$ are measurable, but it can choose what probability it is assigned to each of these sets (i.e., with what probability each action is sampled). Given the strategy $P_i$ of each player, 
 the probability space $(A, \Sigma, Q)$ that gives the probability of each subset of action profiles when each player plays $Q_i$ is precisely the product of $(A_i, \Sigma_i, P_i)$ for all $i \in P$, in which $A := A_1 \times A_2 \times \ldots \times A_n$, $\Sigma := \sigma(\Sigma_1 \otimes \ldots \otimes \Sigma_n)$ (i.e., the $\sigma$-algebra generated by sets of the form $(X_1, \ldots, X_n)$ with $X_i \in \Sigma_i$ for all $i$), and $Q := Q_1 \times \ldots \times Q_n$, which is the unique  measure function such that $Q(X_1, \ldots, X_n) = Q_1(X_1)Q_2(X_2) \ldots Q_n(X_n)$. Fixing $(A, \Sigma, Q)$, the expected utility of a player $i$ is given by the Lebesgue integral $$u_i(Q) := \int_{A} u_i(X) \ dQ.$$

In order to avoid confusion between the strategy profile $\vec{\sigma}$ of Theorems~\ref{thm:main} and \ref{thm:main2} and the $\sigma$ used in the definition of $\sigma$-algebras, in the next sections we will use $Q$ instead of $\vec{\sigma}$ to denote the Nash equilibrium of Theorems~\ref{thm:main} and \ref{thm:main2}.

 \subsubsection{Construction of $\Gamma^Q$}\label{sec:construction-uncountable}

The construction of $\Gamma^Q$ for arbitrary probability spaces is quite similar to the one given in Section~\ref{sec:countable}. In fact, we define $A_i^Q := A_i \cup \left(A_i  \times \bigcup_{j \in P} (\{j\} \times \Sigma_j) \right)$. Intuitively, $i$ can either play an action in $A_i$, or play an action in $A_i$ and simultaneously bet that the action that another player $j$ plays is in some measurable subset $S \in \Sigma_j$. The $\sigma$-algebra of $A_i^Q$ is the one spanned by $\Sigma_i$ and  $\sigma(\Sigma_i \otimes \mathcal{P}(\{j\} \times \Sigma_j))$ for each $j \in P$ (where $\mathcal{P}(S)$ denotes the power set of a set $S$).

We define $U^Q$ exactly in the same way as in Section~\ref{sec:countable}. Note that  $p_{(j,X)}^Q$ is well defined since, by construction, $X$ must belong to $\Sigma_j$. In fact, $p_{(j,X)}^Q = Q_j(X)$.

\subsubsection{High Level Proof}

A similar argument to the one given in Lemma~\ref{lemma:bets-basic} shows that, if players follow strategy profile $Q'$, the expected utility that a player $i$ gets when betting on $(j, S)$ for some $S \in \Sigma_j$ is given by $\tilde{Q}'_j(S) - Q_j(S)$ (where, again, $\tilde{Q}'$ is the projection of $Q'$ onto its first component). Therefore, it is optimal for player $i$ to bet on sets that maximize $\tilde{Q}'_j(S) - Q_j(S)$ with probability $1$. However, it is not clear at first if such a bet always exists, since it could be the case that the supremum of $\{\tilde{Q}'_j(S) - Q_j(S)\}_{S \in \Sigma_j}$ is never attained in $\Sigma_j$. The first part of the proof is precisely showing that such a bet always exists. 

\begin{proposition}\label{prop:max-bet-uncountable}
Let $\Gamma^Q$ be the game constructed in Section~\ref{sec:construction-uncountable} and let $j \in P$ be a player that plays strategy $Q'_j$. Then, for all other players $i \in P$, there exists a maximal and a minimal bet of the form $(j, X)$.
\end{proposition}

The proof of Proposition~\ref{prop:max-bet-uncountable} is somewhat technical and will be shown in Section~\ref{sec:proofs-technical}. Once we have Proposition~\ref{prop:max-bet-uncountable}, the second part of the proof is showing that, among all sets $X$ that maximize or minimize the bets on $j$, there exists one that minimizes $Q(X)$.

\begin{proposition}\label{prop:min-bet-uncountable}
Let $\Gamma^Q$ be the game constructed in Section~\ref{sec:construction-uncountable} and let $j \in P$ be a player that plays strategy $Q'_j$. Let $\Sigma_{Q'_j}^+$ and $\Sigma_{Q'_j}^-$ be the subsets of $\Sigma_j$ that maximize and minimize $\tilde{Q}'_j(S) - Q_j(S)$, respectively. Then, there exists a set $X^+ \in \Sigma_{Q'_j}^+$  and a set $X^-\in \Sigma_{Q'_j}^-$ such that $X^+ \cap X^- = \emptyset$,  $Q(X^+) \le Q(X)$ for all $X \in \Sigma_{Q'_j}^+$ and $Q(X^-) \le Q(X)$ for all $X \in \Sigma_{Q'_j}^-$.
\end{proposition}

The proof of Proposition~\ref{prop:min-bet-uncountable} is quite similar to that of Proposition~\ref{prop:max-bet-uncountable} and is also shown in Section~\ref{sec:proofs-technical}. With these propositions we are ready to give complete the proof of correctness of the construction given in Section~\ref{sec:construction-uncountable}. Suppose that  a subset $S$ of players  with $|S| \ge 2$ defects from $Q_S$ to another strategy profile $Q'_S$ in such a way that (a1) $u_i^Q(Q_{-S}, Q'_S) > u_i^Q(Q)$ for all $i \in S$, and (b1) $u_i^Q(Q_{-S}, Q'_{S \setminus\{i\}}, Q''_i) \le u_i^Q(Q_{-S}, Q'_S)$ for all $i \in S$ and all strategies $Q''_i$ for $i$. Then, there must exist a player $i \in S$ and a subset $X \in \Sigma_i$ such that $\tilde{Q}'_i(X) \not = Q_i(X)$ since, otherwise, all bets would give $0$ expected utility and all players would get at most the expected utility they'd get with $Q$, which contradicts (a). If $\tilde{Q}'_i(X) \not = Q_i(X)$, it means that either $\tilde{Q}'_i(X) - Q_i(X) > 0$ or $\tilde{Q}'_i(\overline{X}) - Q_i(\overline{X}) > 0$, which implies that all other players in $S$ must place bets with probability $1$ (otherwise, it contradicts (b1)). It is easy to check that, if this is the case, there exists a player $j$ such that the expected number of other players that bet on $j$ is at least $1$. Let $X^+$ and $X^-$ be the subsets given in Proposition~\ref{prop:min-bet-uncountable} (for player $j$) and consider a strategy $Q''_j$ bets in the same way as in $Q'_j$ but that intuitively does the following in order to select an action in $A_j$:

\begin{enumerate}
    \item \textbf{Step 1:} $j$ samples an action $a$ according to $Q'_j$.
    \item \textbf{Step 2:} If $a \not \in X^{+}$, $j$ plays $a$.
    \item \textbf{Step 3:} Otherwise, $j$ sequentially samples an action $a$ according to $Q_j$ until $a \in X^{-}$. Then, $j$ plays $a$.
\end{enumerate}

Note that this algorithm always terminates since $\tilde{Q}'_j(X^{-}) - Q_j(X^{-}) < 0$, which means that $Q_j(X^-) > 0$. More precisely, what $j$ is doing is selecting an action in $A_i$ according to the probability measure $\tilde{Q}''_j$ on $(A_j, \Sigma_j)$ defined by $$\tilde{Q}''_j (X) := 
\tilde{Q}'_j(X \setminus X^+) + \frac{\tilde{Q}'_j(X^{+})}{Q_j(X^{-})} Q_j(X \cap X^{-}).$$

It is straightforward to check that $\tilde{Q}''_j(X)$ is indeed a probability measure on $(A_j, \Sigma_j)$:
\begin{itemize}
    \item $\tilde{Q}''_j(\emptyset) = \tilde{Q}'_j(\emptyset) + \frac{\tilde{Q}'_j(X^{+})}{Q_j(X^{-})} Q_j(\emptyset) = 0$.
    \item $\tilde{Q}''_j(A_j) = \tilde{Q}'_j(A_j) - \tilde{Q}'_j(X^+) + \frac{\tilde{Q}'_j(X^{+})}{Q_j(X^{-})} Q_j(X^{-}) = 1$.
    \item If $\{X_i\}_{i \ge 1}$ is a countable sequence of pairwise disjoint subsets in $\Sigma_j$, we have that $\left(\bigcup_{i \ge 1} X_i\right) \setminus X^{+} = \bigcup_{i \ge 1} (X_i \setminus X^{+})$, and $\left(\bigcup_{i \ge 1} X_i\right) \cap X^{-} = \bigcup_{i \ge 1} (X_i \cap X^{-})$. Moreover, the subsets of the form $X_i \setminus X^{+}$ are pairwise disjoint, and so are those of the form $X_i \cap X^{-}$. Therefore, 
    $$\tilde{Q}''_j\left(\bigcup_{i \ge 1} X_i\right) = \sum_{i \ge 1} \tilde{Q}'_j(X_i \setminus X^{+}) + \frac{\tilde{Q}'_j(X^{+})}{Q_j(X^{-})} \sum_{i \ge 1} Q_j(X_i \cap X^{-}) = \sum_{i \ge 1} \tilde{Q}''_j(X_i).$$
\end{itemize}

We show next that $j$ strictly improves its utility by switching from $Q'_j$ to $Q''_j$. Recall that the expected utility $u_j^Q(Q_{-S}, Q'_{S \setminus\{j\}}, Q''_j)$ that $j$ gets when playing strategy $Q''_j$ can be decomposed as  

$$
\begin{array}{lll}
u_j^Q(Q_{-S}, Q'_{S \setminus\{j\}}, Q''_j) & = & u_j(Q_{-S}, \tilde{Q}'_{S \setminus\{j\}}, \tilde{Q}''_j)\\
& + & w_j(Q_{-S}, Q'_{S \setminus\{j\}}, Q''_j)\\
& + & \ell_j(Q_{-S}, Q'_{S \setminus\{j\}}, Q''_j).\\
\end{array}$$

As in previous sections, $w_j(Q_{-S}, Q'_{S \setminus\{j\}}, Q''_j) = w_j(Q_{-S}, Q'_{S})$ since, by construction, $j$ places the same bets in $Q'_j$ and $Q''_j$. We finish with the following lemma, whose proof is given in Section~\ref{sec:proofs-technical}.

\begin{lemma}\label{lemma:uncountable-diff1}
$$
\begin{array}{ll}
u_j(Q_{-S}, \tilde{Q}'_{S \setminus\{j\}}, \tilde{Q}''_j) - u_j(Q_{-S}, \tilde{Q}'_{S}) \le Q'_j(X^+)M\\
\ell_j(Q_{-S}, Q'_{S \setminus\{j\}}, Q''_j) - \ell_j(Q_{-S}, Q'_{S}) > Q'_j(X^+)M.
\end{array}$$
\end{lemma}

Note that this lemma implies that $u_j^Q(Q_{-S}, Q'_{S \setminus\{j\}}, Q''_j) > u_j(Q_{-S}, Q'_{S})$, contradicting the (b1) assumption. This completes the proof of Theorem~\ref{thm:main} part (a). A similar argument can be used to prove part (b), as done in Section~\ref{sec:normal-form-part-b}.

\subsubsection{Proof of the Technical Results}\label{sec:proofs-technical}

In this section, we prove the technical propositions and lemmas given in the previous section.

We begin with Proposition~\ref{prop:max-bet-uncountable}. Its proof follows from the following lemma.

\begin{lemma}\label{lem:global-bounds}
Let $\Sigma$ be a $\sigma$-algebra and let $f:\Sigma \longrightarrow \mathbb{R}$ be a bounded function. Then, $f$ has a global minimum and a global maximum on $\Sigma$.
\end{lemma}

Note that Proposition~\ref{prop:max-bet-uncountable} is an immediate consequence of Lemma~\ref{lem:global-bounds} since the expected utility of a bet is bounded between $1$ and $-1$. To prove Lemma~\ref{lem:global-bounds} we need the following well-known result.

\begin{lemma}\label{lem:union}
Let $\Sigma$ be a $\sigma$-algebra on some set $\Omega$, let $\{X_n\}_{n \ge 1}$ be a sequence of elements of $\Sigma$, and let $f$ be a function from $\Sigma$ to $\mathbb{R}$ that satisfies countable additivity and such that 
\begin{itemize}
\item [(a)] $X_n \subseteq X_{n+1}$ for all $n \in \mathbb{N}$, and
\item [(b)] $\lim_{n \rightarrow \infty} f(X_n) = L$, for some $L \in \mathbb{R}$.
\end{itemize}
Then $f\left(\bigcup_{n \ge 1} X_n\right) = L$.
\end{lemma}

\begin{proof}
Let $Y_n := X_{n+1} \setminus X_{n}$. Then, by (a), $X_1$ and all sets $Y_n$ are pairwise disjoint. Therefore, since $f$ satisfies countable additivity, we have that $$f\left(\bigcup_{i = 1}^n X_i\right) = f(X_n) = f(X_1) + \sum_{i = 1}^{n-1} f(Y_i).$$

Since $\lim_{n \rightarrow \infty} f(X_n) = L$, it means that the sum is convergent and the limit is equal to $L - f(X_1)$. Therefore, $$f\left(\bigcup_{n \ge 1} X_n\right) = f(X_1) + \sum_{i = 1}^{\infty} f(Y_i) = L.$$
\end{proof}

\commentout{
We also need its analogue for intersections.

\begin{lemma}\label{lem-intersection}
Let $\Sigma$ be a $\sigma$-algebra on some set $\Omega$, let $\{X_n\}_{n \ge 1}$ be a sequence of elements of $\Sigma$, and let $f$ be a function from $\Sigma$ to $\mathbb{R}$ that satisfies countable additivity and such that 
\begin{itemize}
\item [(a)] $X_{n+1} \subseteq X_n$ for all $n \in \mathbb{N}$, and
\item [(b)] $\lim_{n \rightarrow \infty} f(X_n) = L$, for some $L \in \mathbb{R}$.
\end{itemize}
Then $f\left(\bigcap_{n \ge 1} X_n\right) = L$.
\end{lemma}

\begin{proof}
The proof is identical to that of Lemma~\ref{lem:union} except that, if we define $Y_n := X_{n} \setminus X_{n+1}$, we have that $$f\left(\bigcap_{i = 1}^n X_i\right) = f(X_n) = f(X_1) - \sum_{i = 1}^{n-1} f(Y_i).$$
\end{proof}
}

To prove Lemma~\ref{lem:global-bounds} we also need the following observation.

\begin{lemma}\label{lem:set-diffs}
Let $\Sigma$ be a $\sigma$-algebra on some set $\Omega$ and let $f$ be a function from $\Sigma$ to $\mathbb{R}$ that satisfies countable additivity and such that $f(X) \le L$ for all $X \in \Sigma$. Let $a_1, a_2, \ldots, a_n$ be a finite sequence of positive real numbers and $X_1, X_2, \ldots, X_n$ be a finite sequence of elements of $\Sigma$ such that $f(X_i) > L - a_i$. Then, $$f\left(\bigcap_{i = 1}^n X_i \right) > L - \sum_{i = 1}^n a_i$$.
\end{lemma}

\begin{proof}
We will only show this result for $n = 2$. The general case follows by induction.

Given two sets $X$ and $Y$ in $\Sigma$ such that $f(X) > L - a$ and $f(Y) > L - b$, let $X' = X \setminus Y$ and let $Y' = Y \setminus X$. Then, we have that 
$$\begin{array}{lll}
f(X') + f(X \cap Y) & >& L - a\\
f(Y') + f(X \cap Y) &>& L - b\\
f(X\cup Y) &\le& L
\end{array}
$$
Writing $X\cup Y$ as $X' \cup Y' \cup (X \cap Y)$, we can add the first two equations and substract the third one to get the desired result.
\end{proof}

With these lemmas we are ready to prove Lemma~\ref{lem:global-bounds}.

\begin{proof}[Proof of Lemma~\ref{lem:global-bounds}]
We will only show that $f$ has a global maximum since, if $X$ is a global maximum of $f$, then $\overline{X}$ is a global minimum (and viceversa). Let $L$ be the supremum of $\{f(X)\}_{X \in \Sigma}$ and let $X_1, X_2, \ldots$ be a sequence of elements of $\Sigma$ such that $f(X_n) > L - \frac{1}{2^n}$. Let $Y_m^n = \bigcap_{i = m}^n X_i$. By Lemma~\ref{lem:set-diffs} we have that $f(Y_m^n) > L - \sum_{i = m}^n \frac{1}{2^i} = L - (\frac{1}{2^{m-1}} - \frac{1}{2^n})$.

If we take the limit when $n \rightarrow \infty$, we have that $f(Y_m^\infty) > L - \frac{1}{2^{m-1}}$ (note that $Y_m^\infty \in \Sigma$ since $\Sigma$ is closed under countable intersections). By construction, $Y_m^{\infty} \subseteq Y_{m+1}^{\infty}$ for all $m \in \mathbb{N}$. Moreover, $$\lim_{m \rightarrow \infty} f(Y_m^{\infty}) = \lim_{m \rightarrow \infty} L - \frac{1}{2^{m-1}} = L.$$

By Lemma~\ref{lem:union}, this means that $$f\left(\bigcup_{m = 1}^{\infty} Y_m^{\infty}\right) = L.$$

This completes the proof.
\end{proof}

The proof of Proposition~\ref{prop:min-bet-uncountable} is analogous to that of Proposition~\ref{prop:max-bet-uncountable}. The key observation is that the set of maximal and the set of minimal bets are closed under countable unions and intersections, which means that we can do an analogue of the proof of Lemma~\ref{lem:global-bounds} restricted to $\Sigma_{Q_j^+}$ and to $\Sigma_{Q_j^-}$. To prove this, note that if $f$ is a countably additive bounded function defined on a $\sigma$-algebra $\Sigma$, we have that $f(X) = f(X\cap Y) + f(X \setminus Y)$ and $f(Y) = f(X \cap Y) + f(Y \setminus X)$ for any two sets $X,Y \in \Sigma$. If Both $X$ and $Y$ maximize $f$, then $ f(Y \setminus X) = f(X \setminus Y) = 0$. Otherwise, $f(X \cap Y)$ or $f(X \cup Y)$ would be strictly greater than $f(X)$. This implies that $f(X \cap Y) = f(X \cup Y) = f(X) = f(Y)$.

It is important to note that if we restrict the proof of Lemma~\ref{lem:global-bounds} to  $\Sigma_{Q_j^+}$ and to $\Sigma_{Q_j^-}$, finding the minimum does not reduce to finding the maximum (as we did in Lemma~\ref{lem:global-bounds}) since $\Sigma_{Q_j^+}$ and $\Sigma_{Q_j^-}$ are not closed under complements. However, we can reproduce an analogous proof using analogues to Lemma~\ref{lem:union} and Lemma~\ref{lem:set-diffs} for intersections and infimums, respectively. Finally, note that we can assume w.l.o.g. that the sets $X^+$ and $X^-$ in Proposition~\ref{prop:min-bet-uncountable} are disjoint since, if $X^+ \cap X^- \not = \emptyset$, then we can use the same argument as above to show that $Q'_j(X^ \cap X^-) = Q_j(X^ \cap X^-) = 0$. If this happens, we can take $X^{+} \setminus X^{-}$ and $X^{-} \setminus X^+$ instead of $X^+$ and $X^-$.

It remains to show Lemma~\ref{lemma:uncountable-diff1}.

\begin{proof}[Proof of Lemma~\ref{lemma:uncountable-diff1}]
We prove each equation separately.

\textbf{First equation: $u_j(Q_{-S}, \tilde{Q}'_{S \setminus\{j\}}, \tilde{Q}''_j) - u_j(Q_{-S}, \tilde{Q}'_{S}) \le Q'_j(X^+)M$.}

For simplicity, let $R$ denote the strategy profile where players in $S$ play $\tilde{Q}'_S$ and the rest of the players play $Q_{-S}$.

By definition, $$u_j(R_{-j}, \tilde{Q}''_j) = \int_{A} u_j(X) \ d(R_{-j}, \tilde{Q}''_j).$$

We can split $A$ into $\Omega_1 := X^+ \times A_{-i}$ and $\Omega_2 := (A_i \setminus X^+) \times A_{-i}$ to get $$u_j(R_{-j}, \tilde{Q}''_j) = \int_{\Omega_1} u_j(X) \ d(R_{-j}, \tilde{Q}''_j) + \int_{\Omega_2} u_j(X) \ d(R_{-j}, \tilde{Q}''_j).$$

By construction, $Q_j''(X) = 0$ for all $X \subseteq X^{+}$, which means that $\int_{\Omega_1} u_j(X) \ d(R_{-j}, \tilde{Q}''_j) = 0$. Additionally, if we define $\Omega_3 := X^{-} \times A_{-i}$, we have that $$\int_{\Omega_2} u_j(X) \ d(R_{-j}, \tilde{Q}''_j) = \int_{\Omega_2} u_j(X) \ dR + \frac{\tilde{Q}'_j(X^{+})}{Q_j(X^{-})} \int_{\Omega_3} u_j(X)  \ d(R_{-j}, Q_j).$$

Note that 
$
u_j(R) = \int_{\Omega_2} u_j(X) \ dR + \int_{\Omega_1} u_j(X) \ dR
$, which implies that 

$$u_j(R_{-j}, \tilde{Q}''_j) - u_j(R) = \frac{\tilde{Q}'_j(X^{+})}{Q_j(X^{-})} \int_{\Omega_3} u_j(X)  \ d(R_{-j}, Q_j) - \int_{\Omega_1} u_j(X) \ dR.$$

 If $M_{\max}$ is the supremum of $u_j$ and $M_{\min}$ is its infimum, we can bound $u_j(R_{-j}, \tilde{Q}''_j) - u_j(R)$ by 
$$
u_j(R_{-j}, \tilde{Q}''_j) - u_j(R) \le \frac{\tilde{Q}'_j(X^{+})}{Q_j(X^{-})} M_{\max} \cdot (R_{-j}, Q_j)(\Omega_3) - M_{\min} \cdot R(\Omega_1).$$

Note, however, that $(R_{-j}, Q_j)(\Omega_3) = Q_j(X^-)$ and $R(\Omega_1) = \tilde{Q}'_j(X^{+})$. Therefore, 
$$
u_j(R_{-j}, \tilde{Q}''_j) - u_j(R) \le  Q'_j(X^+)(M_{\max} - M_{\min}),
$$

as desired.

\textbf{Second equation: $\ell_j(Q_{-S}, Q'_{S \setminus\{j\}}, Q''_j) - \ell_j(Q_{-S}, Q'_{S}) > Q'_j(X^+) \cdot M$.}

The second equation follows from the observation we made in the proof of Proposition~\ref{prop:min-bet-uncountable}: if $(j,X)$ is a maximal bet on $j$, then $Q_j(X \setminus X^{+}) = Q_j'(X \setminus X^{+}) = 0$. Since $\tilde{Q}''_j(X^+) = 0$, it follows that $\ell_j(Q_{-S}, Q'_{S \setminus\{j\}}, Q''_j) = 0$. Thus, it remains to show that $\ell_j(Q_{-S}, Q'_{S}) < - Q'_j(X^+)\cdot M$. This follows from the same observation since then $\ell_j(Q_{-S}, Q'_{S}) = \mathbb{E}[\mbox{\# of other players betting on j}] \cdot Q'_j(X^+) \cdot (-M-1)$. Since we assumed that at least one other player bets on $j$ in expectation, it follows that $\ell_j(Q_{-S}, Q'_{S}) < - Q'_j(X^{+})\cdot M$.
\end{proof}

\section{Conclusion and Open Problems}\label{sec:conclusion}

This paper shows that if a mechanism designer has the power to add new actions to a normal-form or Bayesian game, it can shield a given Nash equilibrium from collusion and from equilibrium selection issues. More precisely, given a normal-form (resp., Bayesian) game $\Gamma$ and a Nash equilibrium $\vec{\sigma}$ of $\Gamma$, we showed that exists an extension $\Gamma^{\vec{\sigma}}$ of $\Gamma$ such that (a) $\vec{\sigma}$ is a semi-strong Nash equilibrium of $\Gamma$ and (b) $\vec{\sigma}$ Pareto-dominates (resp., quasi Pareto-dominates) all other equilibria of $\Gamma^{\vec{\sigma}}$. This result opens the door for three important follow-up problems.

\begin{itemize}
\item \textbf{Can we also get Pareto-dominance for arbitrary Bayesian games?} We proved that the construction of $\Gamma^{\vec{\sigma}}$ given in Section~\ref{sec:bayesian-construction} satisfies that $\vec{\sigma}$ quasi Pareto-dominates all other Nash equilibria of $\Gamma^{\vec{\sigma}}$. We conjecture that, in fact, $\vec{\sigma}$ also Pareto-dominates all other Nash equilibria in this construction.
\item \textbf{Minimizing promises.} As discussed in earlier sections, we can view the newly added actions in $\Gamma^{\vec{\sigma}}$ as \emph{promises} made by the mechanism designer to the players. In practice, if players are rational and have full information, these newly actions should never be played since every Nash equilibrium that includes them is Pareto-dominated (or quasi Pareto-dominated) by $\vec{\sigma}$. Thus, in theory, the mechanism designer can simply bluff about the existence of these new actions (e.g., by promising to play certain amount to a given player). Nevertheless, it is still an open question to see exactly how much utility must the mechanism designer promise in order to construct a game $\Gamma^{\vec{\sigma}}$ that satisfies the properties of Theorem~\ref{thm:main} and \ref{thm:main2}.
\item \textbf{Minimum size of $\Gamma^{\vec{\sigma}}$.} Theorems~\ref{thm:main} and \ref{thm:main2} give upper bounds on the minimum number of actions that must be added in order to construct $\Gamma^{\vec{\sigma}}$. It is still an open problem to verify if there exists a smaller construction.
\end{itemize}

\bibliographystyle{plain}
\bibliography{bibfile}

\begin{thebibliography}{10}

\bibitem{abraham2011distributed}
Ittai Abraham, Lorenzo Alvisi, and Joseph~Y Halpern.
\newblock Distributed computing meets game theory: combining insights from two
  fields.
\newblock {\em Acm Sigact News}, 42(2):69--76, 2011.

\bibitem{abraham2019implementing}
Ittai Abraham, Danny Dolev, Ivan Geffner, and Joseph~Y Halpern.
\newblock Implementing mediators with asynchronous cheap talk.
\newblock In {\em Proceedings of the 2019 ACM Symposium on Principles of
  Distributed Computing}, pages 501--510, 2019.

\bibitem{abraham2006distributed}
Ittai Abraham, Danny Dolev, Rica Gonen, and Joseph~Y Halpern.
\newblock Distributed computing meets game theory: robust mechanisms for
  rational secret sharing and multiparty computation.
\newblock In {\em Proceedings of the twenty-fifth annual ACM symposium on
  Principles of distributed computing}, pages 53--62, 2006.

\bibitem{abraham2008lower}
Ittai Abraham, Danny Dolev, and Joseph~Y Halpern.
\newblock Lower bounds on implementing robust and resilient mediators.
\newblock In {\em Theory of Cryptography: Fifth Theory of Cryptography
  Conference, TCC 2008, New York, USA, March 19-21, 2008. Proceedings 5}, pages
  302--319. Springer, 2008.

\bibitem{aiyer2005bar}
Amitanand~S Aiyer, Lorenzo Alvisi, Allen Clement, Mike Dahlin, Jean-Philippe
  Martin, and Carl Porth.
\newblock Bar fault tolerance for cooperative services.
\newblock In {\em Proceedings of the twentieth ACM symposium on Operating
  systems principles}, pages 45--58, 2005.

\bibitem{aumann1959acceptable}
Robert~J Aumann.
\newblock Acceptable points in general cooperative n-person games.
\newblock {\em Contributions to the Theory of Games}, 4:287--324, 1959.

\bibitem{aumann1974subjectivity}
Robert~J Aumann.
\newblock Subjectivity and correlation in randomized strategies.
\newblock {\em Journal of mathematical Economics}, 1(1):67--96, 1974.

\bibitem{bgw88}
M.~Ben-Or, S.~Goldwasser, and A.~Wigderson.
\newblock Completeness theorems for non-cryptographic fault-tolerant
  distributed computation.
\newblock In {\em Proceedings of the 20th ACM Symposium on Theory of
  Computing}, pages 1--10, 1988.

\bibitem{ben1983another}
Michael Ben-Or.
\newblock Another advantage of free choice (extended abstract) completely
  asynchronous agreement protocols.
\newblock In {\em Proceedings of the second annual ACM symposium on Principles
  of distributed computing}, pages 27--30, 1983.

\bibitem{ben1993asynchronous}
Michael Ben-Or, Ran Canetti, and Oded Goldreich.
\newblock Asynchronous secure computation.
\newblock In {\em Proceedings of the twenty-fifth annual ACM symposium on
  Theory of computing}, pages 52--61, 1993.

\bibitem{bernheim1987coalition}
B~Douglas Bernheim, Bezalel Peleg, and Michael~D Whinston.
\newblock Coalition-proof nash equilibria i. concepts.
\newblock {\em Journal of economic theory}, 42(1):1--12, 1987.

\bibitem{bracha1984asynchronous}
Gabriel Bracha.
\newblock An asynchronous [(n-1)/3]-resilient consensus protocol.
\newblock In {\em Proceedings of the third annual ACM symposium on Principles
  of distributed computing}, pages 154--162, 1984.

\bibitem{bracha1985asynchronous}
Gabriel Bracha and Sam Toueg.
\newblock Asynchronous consensus and broadcast protocols.
\newblock {\em Journal of the ACM (JACM)}, 32(4):824--840, 1985.

\bibitem{forges1986approach}
Francoise Forges.
\newblock An approach to communication equilibria.
\newblock {\em Econometrica: Journal of the Econometric Society}, pages
  1375--1385, 1986.

\bibitem{geffner2021security}
Ivan Geffner and Joseph~Y Halpern.
\newblock Security in asynchronous interactive systems.
\newblock In {\em International Symposium on Stabilizing, Safety, and Security
  of Distributed Systems}, pages 123--140. Springer, 2021.

\bibitem{geffner2023communication}
Ivan Geffner and Joseph~Y Halpern.
\newblock Communication games, sequential equilibrium, and mediators.
\newblock {\em arXiv preprint, arXiv:2309.14618}, 2023.

\bibitem{geffner2023lower}
Ivan Geffner and Joseph~Y Halpern.
\newblock Lower bounds on implementing mediators in asynchronous systems with
  rational and malicious agents.
\newblock {\em Journal of the ACM}, 70(2):1--21, 2023.

\bibitem{monderer2003k}
Dov Monderer and Moshe Tennenholtz.
\newblock k-implementation.
\newblock In {\em Proceedings of the 4th ACM conference on Electronic
  Commerce}, pages 19--28, 2003.

\bibitem{monderer2009strong}
Dov Monderer and Moshe Tennenholtz.
\newblock Strong mediated equilibrium.
\newblock {\em Artificial Intelligence}, 173(1):180--195, 2009.

\bibitem{oesterheld2022similarity}
Caspar Oesterheld, Johannes Treutlein, Roger Grosse, Vincent Conitzer, and
  Jakob Foerster.
\newblock Similarity-based cooperation.
\newblock {\em arXiv preprint, arXiv:2211.14468}, 2022.

\bibitem{penn2005congestion}
Michal Penn, Maria Polukarov, and Moshe Tennenholtz.
\newblock Congestion games with failures.
\newblock In {\em Proceedings of the 6th ACM Conference on Electronic
  Commerce}, pages 259--268, 2005.

\bibitem{rabin1983randomized}
Michael~O Rabin.
\newblock Randomized byzantine generals.
\newblock In {\em 24th annual symposium on foundations of computer science
  (sfcs 1983)}, pages 403--409. IEEE, 1983.

\bibitem{ranchal2021rational}
Alejandro Ranchal-Pedrosa and Vincent Gramoli.
\newblock Rational agreement in the presence of crash faults.
\newblock In {\em 2021 IEEE International Conference on Blockchain}, pages
  470--475. IEEE, 2021.

\bibitem{tennenholtz2004program}
Moshe Tennenholtz.
\newblock Program equilibrium.
\newblock {\em Games and Economic Behavior}, 49(2):363--373, 2004.

\bibitem{vilacca2011n}
Xavier Vila{\c{c}}a, Joao Leitao, Miguel Correia, and Lu{\'\i}s Rodrigues.
\newblock N-party bar transfer.
\newblock In {\em Proceedings of the 2011 Conference on Principles of
  Distributed Systems}, pages 392--408. Springer, 2011.

\end{thebibliography}

\end{document}